\documentclass[twoside]{amsart}
\usepackage{latexsym}
\usepackage{amssymb,amsmath,amsopn}
\usepackage[dvips]{graphicx}   
\usepackage{color,epsfig}      

               {\begin{list}{}{\leftmargin#1\rightmargin#2}\item{}}%
               {\end{list}}
 
\newtheorem{thm}{Theorem}
\newtheorem{lem}{Lemma}
\newtheorem{prop}{Proposition}
\theoremstyle{definition}
\newtheorem{defn}{Definition}

\newtheorem{rem}{Remark}

\newtheorem{cor}{Corollary}

\newcommand{\dg}{\color{black}}

\def\bea{\begin{eqnarray}}
\def\eea{\end{eqnarray}}

\usepackage[normalem]{ulem}

\author[G\'abor Domokos]{G\'abor Domokos}
\address{G\'abor Domokos, HUN-REN-BME Morphodynamics Research Group and Dept. of Morphology and Geometric Modeling, Budapest University of Technology,
M\H uegyetem rakpart 1-3., Budapest, Hungary, 1111}
\email{domokos@iit.bme.hu}

\begin{document}
\title[Soft cells, Tubular Tilings and the Hidden Phases in Binary Mixtures]{Soft cells, Tubular Tilings and the Hidden Phases in Binary Mixtures}


\begin{abstract}Biological and physical systems ranging from Fermi surfaces and skeletal structures to reaction--diffusion patterns and cosmological models may be viewed as binary mixtures in which a smooth interface separates two complementary phases. While the interface is often directly observable, the topology of one of the phases may remain hidden. To study such systems, we introduce tubular tilings, a geometric framework for discretizing binary mixtures on smooth manifolds of arbitrary dimension and topology. Under mild regularity conditions, we prove that tubular tilings satisfy global Euler balance laws relating the topology of the ambient manifold, the discretized phases, and their interfaces. These balance laws provide a practical inference principle: topological information about a hidden phase can be recovered from the observable phase and the geometry of the separating interface. We further show that, in dimensions $d>2$, tubular tilings form a subclass of soft tilings, the recently discovered class of corner-free tessellations. Applications to Fermi surfaces and cosmological shell decompositions illustrate how the theory can be used to extract otherwise inaccessible topological information about complex geometric structures.
\end{abstract}
\maketitle
\tableofcontents

\maketitle

\section{Introduction}
Tilings of space are among the most striking geometric patterns in nature and art.
In mathematical language, they are coverings of a smooth, boundaryless manifold $M^d$ by compact sets (tiles) without gaps or overlaps.
Beyond their role as models of physical and biological structures, tilings also provide an intuitive bridge between continuous space and discrete description.

Many natural systems may be interpreted as binary mixtures: the ambient manifold is partitioned into two complementary phases separated by a smooth interface. Examples range from biological tissues and porous materials to reaction--diffusion patterns, triply periodic minimal surfaces (TPMS) and cosmological decompositions of space. In many applications the separating interface is directly observed, whereas only one of the two complementary phases is naturally interpreted as the physical object. Throughout the paper we refer to the other complementary phase as the \emph{hidden phase}, although the interface determines both complementary phases equally. While the interface itself is often the primary object of study, many applications require a finite geometric representation of the two phases. A prominent example is provided by TPMS, whose complementary labyrinths admit natural discretizations into unit cells. This motivates the search for natural discretizations of binary mixtures.

In this paper we introduce such discretizations, which we call tubular tilings, and derive topological balance laws governing their structure. These balance laws allow the topology of the hidden phase to be inferred from the observed phase together with the geometry of the separating interface. The resulting framework connects smooth interfaces, soft tilings and topological invariants of discretized binary mixtures.

\subsection{Tubular tilings}

The motivation for tubular tilings originates from the observation that many binary mixtures naturally admit finite geometric representations.  Discretizations analogous to TPMS structures appear in cellular materials, porous media, biological tissues and reaction--diffusion systems.
To capture these structures in a unified setting, we introduce \emph{tubular tilings}.

Informally, one starts with a tiling $\mathcal T$ of a connected, complete Riemannian manifold $M$ and assigns each tile one of two labels, $A$ or $B$, producing two complementary spatial components.
Interfaces shared by tiles within the same component form the \emph{internal interface} of a tile, whereas interfaces shared with tiles of the opposite component form its \emph{external interface}.

A tiling is called \emph{tubular} if there exists a smooth embedded hypersurface $T$ such that the external interfaces collectively tile $T$
(tiling condition), while the internal interfaces of tiles remain disjoint unions of faces (tubularity condition).
In this case, we call $T$ a \emph{tubular (hyper)surface}.
A rigorous definition is given in subsection \ref{ss:tt},  Definition \ref{defn:tubulartilings}.

Our goal is to derive a balance law for the average Euler characteristics \(\chi_A,\chi_B\) of the \(A\)- and \(B\)-tiles, respectively, and the average Euler characteristics \(\chi_{\bar A},\chi_{\bar B}\) of their corresponding internal interfaces. To justify the finite-volume argument and the passage to the infinite-volume limit, we impose mild regularity conditions ensuring that boundary errors and finite-volume topological contributions become negligible after normalization; tubular tilings satisfying these conditions are called \emph{topologically balanced}. A precise definition is given in subsection \ref{ss:tt}. Topological balance is automatic on compact manifolds and also holds in our principal examples.

Figure~\ref{fig:tub} illustrates the concept. Figure~\ref{fig:tub}(a) shows the frontal view of a brick wall interpreted as a tubular tiling, while Figure~\ref{fig:tub}(b) shows the analogous construction for the Schwarz~D surface. Although geometrically very different, both examples share the same topological structure: a smooth interface separating two complementary phases whose discretization defines a tubular tiling.

\begin{figure*}
\begin{center}
\includegraphics[width=\linewidth]{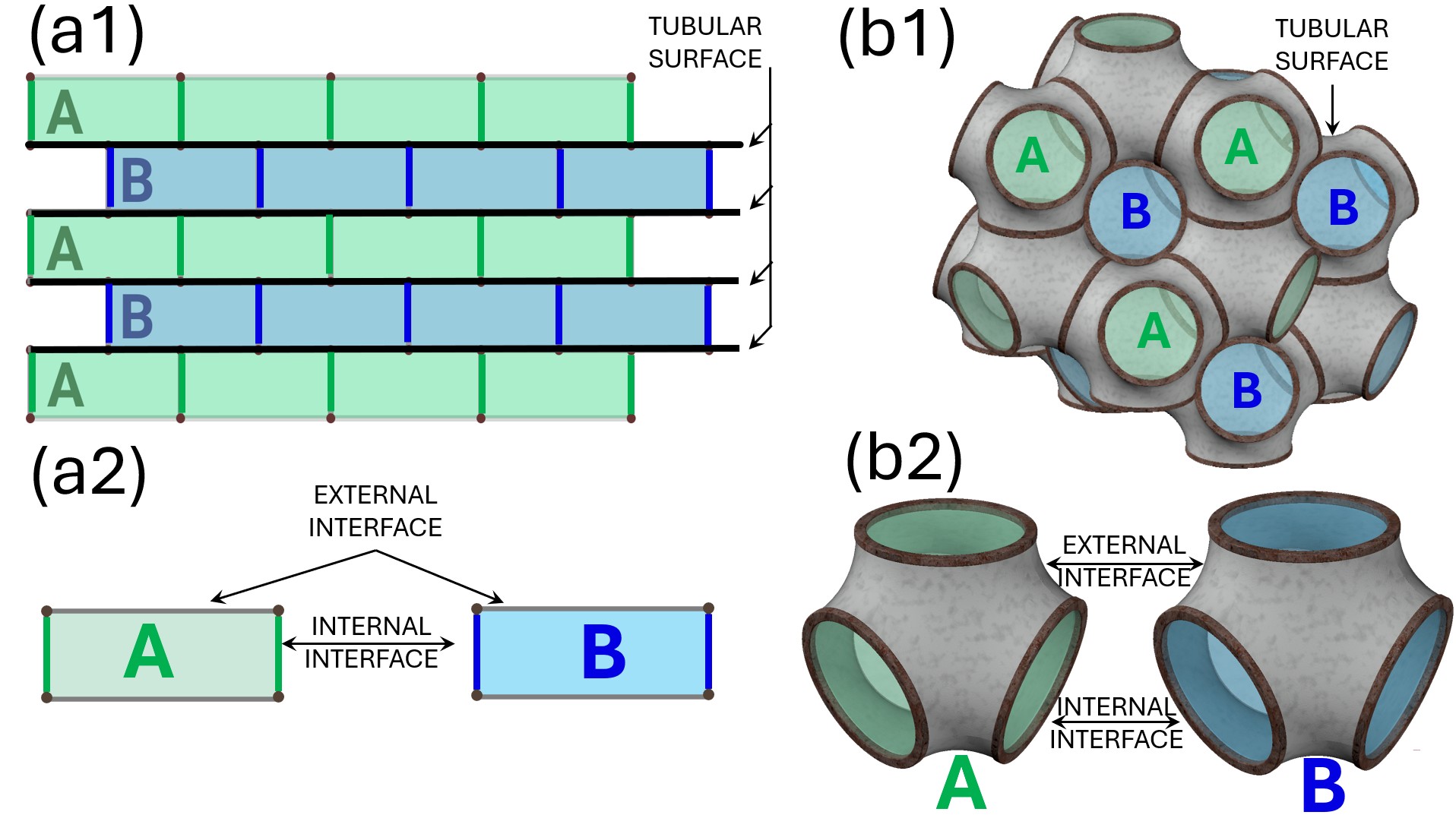}
\caption{Tubular tilings: basic concepts. (a1) Frontal view of brick wall laid in bond. Tubular domains $A$ and $B$ correspond to alternating layers. Tubular surface $T$ is the union of horizontal lines separating layers. Tubular tiles are rectangular views of bricks. (a2) External ($A/B$) and internal ($A/A$, $B/B$) interfaces on brick wall. (b1) Schwarz D surface interpreted as tubular tiling. Tubular surface (grey), $A$ and $B$ tiles indicated by green and blue semi-transparent interiors. (b2) External ($A/B$) interfaces (grey) and internal interfaces ($A/A$ green, $B/B$ blue) on tubular tiles of the Schwarz D surface.}\label{fig:tub}
\end{center}
\end{figure*}

\begin{figure}
    \centering
    \includegraphics[width=\linewidth]{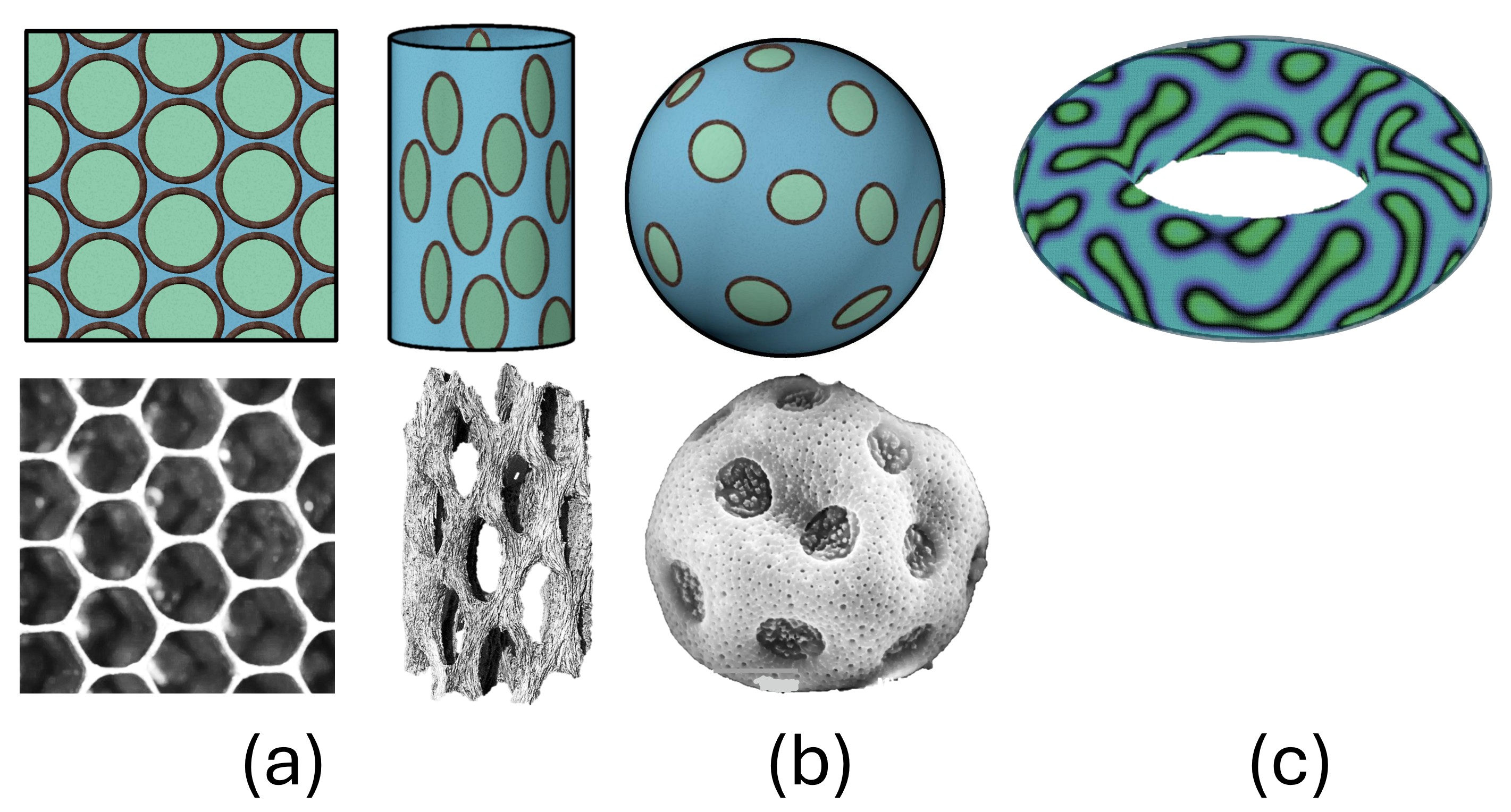}
    \caption{Examples of tubular tilings on two dimensional manifolds. (a) $M^d=T^2$ (flat torus): the honeycomb and $M^d=S^1\times \mathbb I^1$ (periodic cylinder): cactus skeleton (b) $M^d=S^2$ (sphere): pollen. (c) $M^d=T^2$ (non-flat torus): Turing patterns computed on a torus.  }
    \label{fig:2Dpatterns}
\end{figure}
\begin{figure*}
    \begin{center}
    \includegraphics[width=\linewidth]{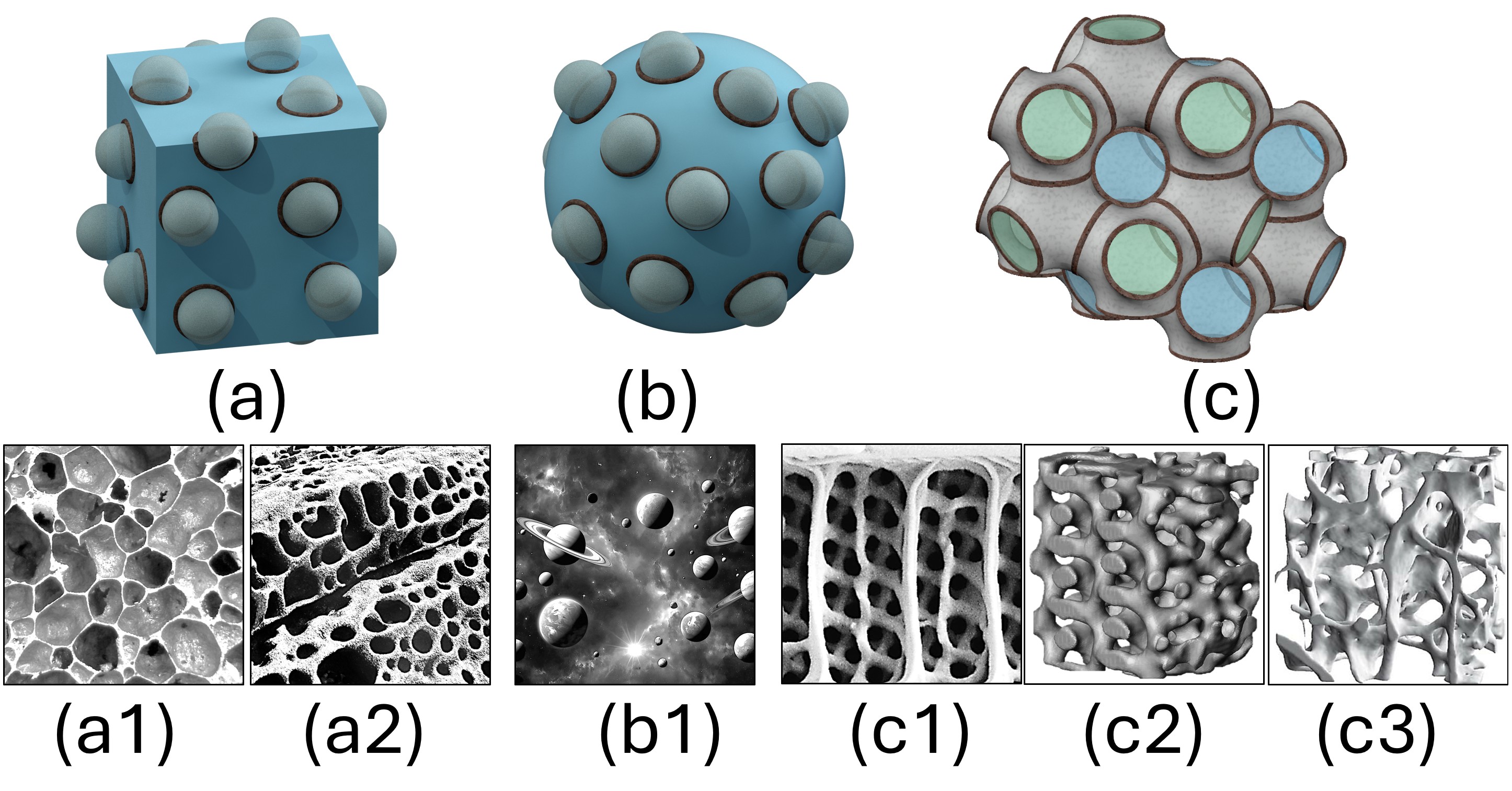}
    \caption{Examples of tubular tilings on three dimensional manifolds. (a) $M^d=T^3$ (flat torus): (a1) metal foam (a2) tafoni rock  (b) $M^d=S^3$ (b1) Thick shell decomposition of the Robertson-Walker universe. Tubular surface is the union of planetary surfaces. (c) $M^d=R^3$ (Euclidean space): (c1) butterfly wing/gyroid structure (c2) diblock copolymer (c3) human bone structure }
    \label{fig:3Dpatterns}
    \end{center}
\end{figure*}

Tubular tilings occur naturally in a wide range of systems. Examples on two-dimensional manifolds are shown in Figure~\ref{fig:2Dpatterns}, while Figure~\ref{fig:3Dpatterns} illustrates examples in three dimensions, including cellular structures, reaction--diffusion patterns, porous materials and cosmological decompositions. The same geometric structure that underlies tubular tilings also endows them with remarkable regularity properties, making it possible to derive global Euler balance laws and relate them to soft tilings.

\subsection{The main result}

Tubular tilings provide a framework for studying the topology of discretized binary mixtures on smooth, boundaryless manifolds of finite topological type.
Under the tiling and tubularity conditions, the topology of the ambient manifold constrains the topology of the discretization.

Our main result, proved in Appendix \ref{sec:thm_euler}, is an Euler-type balance law relating the topology of tubular cells, their internal interfaces and the ambient manifold.

\begin{thm}\label{thm:eulerD}
Let $M^d$ be a connected, complete Riemannian $d$--manifold without boundary and of finite
topological type, and let $\mathcal{T}$ be a tubular, topologically balanced tiling of $M^d$,
consisting of $A$--tiles and $B$--tiles.

Assume that the average Euler characteristics
$\chi_{\mathbf A},\chi_{\mathbf B}$ of the $A$-- and $B$--tiles,
respectively, and the average Euler characteristics
$\chi_{\bar A},\chi_{\bar B}$ of their corresponding internal
interfaces exist. Let $N$ denote the total number of tiles and assume that
the respective relative frequencies $p_A,p_B$ of $A$-- and $B$--tiles also exist.
Then
\begin{equation}\label{eq:thm0D}
p_A(2\chi_{\mathbf A}-\chi_{\bar A})+
(-1)^d p_B(2\chi_{\mathbf B}-\chi_{\bar B})=
\frac{\chi(M^d)}{N}
\bigl((-1)^d+1\bigr).
\end{equation}
\end{thm}

We immediately observe that the right-hand side of (\ref{eq:thm0D}) vanishes whenever either $M^d$ is noncompact (in which case $N$ is infinite) or $d$ is odd.
In these cases the balance law reduces to a homogeneous relation.

Theorem~\ref{thm:eulerD} also recovers balance relations for classical convex mosaics as a special case.
In particular, for a normal, balanced convex tiling of $\mathbb R^3$, let $F$ and $V$ denote the average numbers of faces and vertices of cells, and let $f$ and $v$ denote the corresponding quantities for vertex polyhedra.
We prove (Corollary~\ref{cor:tub} in Section \ref{sec:app}) that $$
V(2-v)-f(2-F)=0.$$

Beyond its geometric consequences, Theorem~\ref{thm:eulerD} provides a topological tool for studying binary mixtures through their discretizations.
In many physical systems one phase is directly observable, whereas the complementary phase is accessible only indirectly or incompletely and we will call this the hidden phase. The theorem enables one to infer topological information
about a hidden phase from
geometric data associated with the observable phase and
their common interface.

\subsection{Relation to soft tilings}\label{ss:isoft}

The concept of tubular tilings grew out of the theory of
\emph{soft tilings} introduced in \cite{softcells1}. Soft tilings
were originally motivated by the study of space-filling structures
with a minimal number of sharp corners and by geometric models of
natural forms such as the chambered shell of the Nautilus
\cite{softcells1,Galo2024}, biological tissues
\cite{theis2025cellmet} and corals \cite{yuval2025}.
The Edge Bending (EB) algorithm \cite{softcells1, ambrus} generates such
tilings from convex mosaics, while its extension in \cite{softcells2}
revealed a broader family of soft cells, including the unit cells of
classical TPMS (Figure~\ref{fig:0}).
\begin{figure}
\centering
\includegraphics[width=\linewidth]{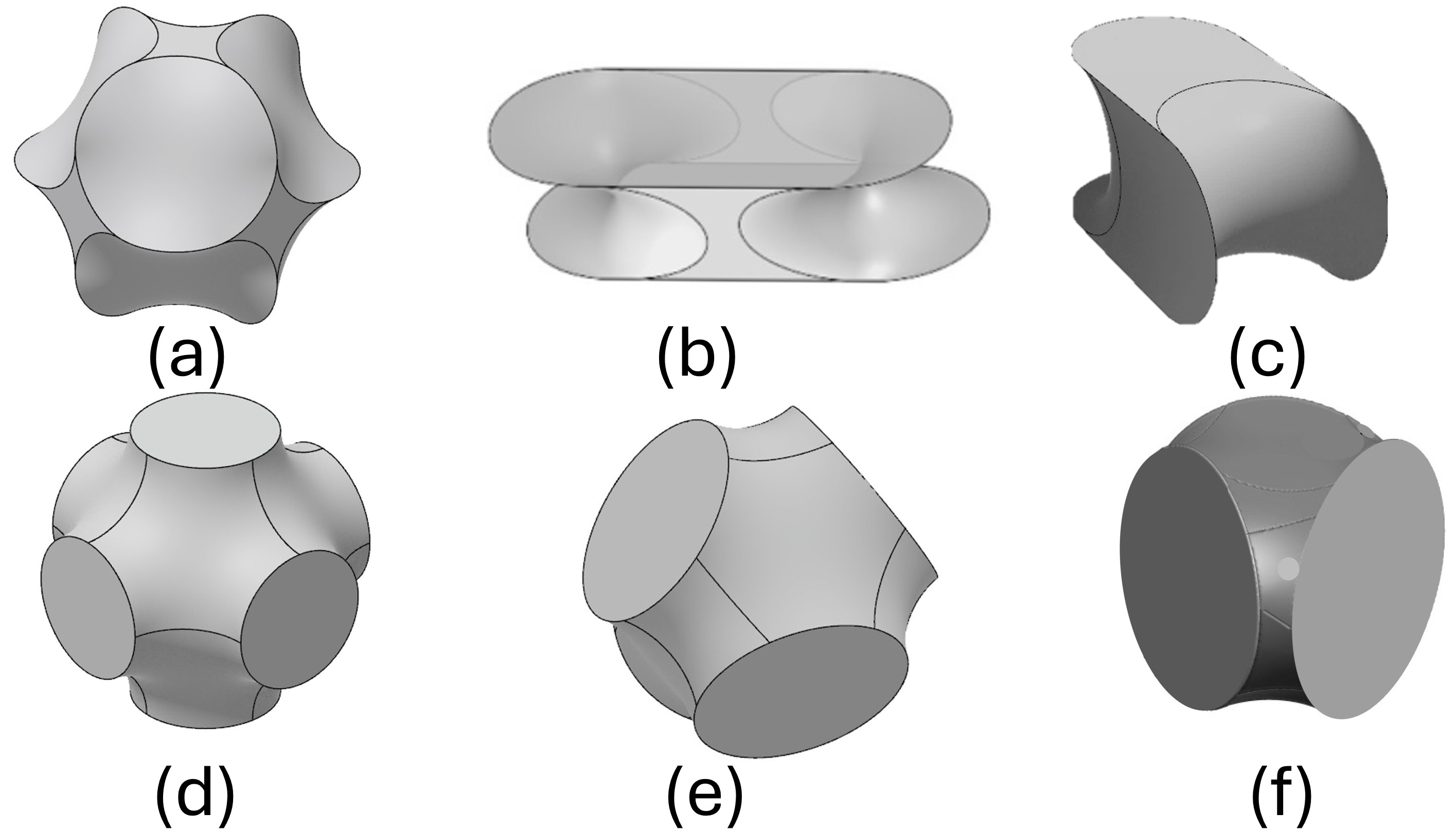}
\caption{Soft cells. First row: soft cells obtained by the edge bending (EB) algorithm. (a) The f2 soft cell obtained from the truncated octahedron. (b) Soft cell with $Z_2 \times Z_2 \times Z_2$ symmetry, obtained from the cube. (c) Soft cell with $Z_2 \times Z_2$ symmetry, obtained from the cube. Second row: soft cells associated with classical triply periodic minimal surfaces (TPMS). (d) Schwarz P unit cell. (e) Schwarz D unit cell. (f) Gyroid unit cell.}
\label{fig:0}
\end{figure}

We show (subsection \ref{ss:soft}, Proposition \ref{prop:soft}), that tubular tilings possess strong
regularity properties. In particular, when $d>2$ they contain no
sharp corners and therefore belong to the class of soft tilings.
From this perspective, softness emerges naturally from the topology
of discretized binary mixtures rather than solely from local
geometric constraints.

To compare tubular tilings with earlier constructions, we adopt a
generalized notion of softness (subsection \ref{ss:soft}, 
Definition \ref{defn:soft2}). A tile is called $k$-soft if the
highest codimension of its smooth boundary strata is $k$. This
definition admits multiply connected tiles, extends naturally to
arbitrary dimensions and, in dimension three, identifies
corner-free tilings with 2-soft tilings.

\subsection{Strategy of proof}

In Appendix \ref{sec:basic} we give the formal definitions of tubular tilings and soft cells, while Appendix \ref{sec:thm_euler} contains the proof of Theorem~\ref{thm:eulerD}.
Below we motivate both the definitions and the strategy of the proof.

Tubular tilings may be introduced in two complementary ways: either by first defining a tubular surface and subsequently introducing a discretization, or conversely by associating a smooth surface with an existing tiling.
In the Introduction we followed the latter approach, and in Appendix \ref{sec:basic} we give a formal definition (Definition \ref{defn:tubulartilings}) in the same spirit.

In the proof, rather than refining tubular tilings into CW complexes, we work directly with their natural geometric structure using inclusion--exclusion for Euler characteristic.

Convex tilings by polyhedral cells naturally define CW complexes, and their topology is captured by the Euler--Poincaré formula
\begin{equation}\label{eq:EP}
\chi(M^d)=\sum_{k=0}^d (-1)^k c_k,
\end{equation}
which computes the Euler characteristic of the ambient manifold as an alternating sum of the numbers of $k$--cells.

Tubular tilings, however, are not naturally CW complexes. A tubular tile is a compact $d$--manifold with piecewise smooth boundary whose geometric structure consists of $(d-1)$--dimensional strata (the internal interface, given by disjoint cross-sectional manifolds, and the external interface, consisting of disjoint smooth tubular patches) together with $(d-2)$--dimensional strata (their smooth boundaries). Such cells generally have nontrivial topology, and their Euler characteristic is not determined solely by dimension. Converting this geometric decomposition into a CW complex would require introducing numerous auxiliary strata and subdividing the geometric cells into contractible pieces.

While such a refinement would make the Euler--Poincaré formula~(\ref{eq:EP}) formally applicable, it has two disadvantages for the present purposes. First, it introduces substantial auxiliary combinatorial detail unrelated to the intrinsic geometric structure of tubular cells. Second, it obscures essential geometric information by replacing multiply connected geometric cells with collections of disks. In particular, adjacency relations and disk valencies, which play a central role in the balance laws, are not naturally encoded at the CW level without additional bookkeeping.

For this reason, the Euler balance laws are formulated directly at the level of the natural geometric strata of the tiling using inclusion--exclusion \cite{Hatcher}:
\begin{equation}\label{eq:add}
\chi(X\cup Y)=\chi(X)+\chi(Y)-\chi(X\cap Y).
\end{equation}
Conceptually, this may be viewed as a geometry-adapted compression of the full CW Euler bookkeeping: auxiliary strata are suppressed, while the gluing relations between tubular cells and separating disks are kept explicit.


\section{Geometric constructions and physical applications}\label{sec:app}

Theorem~\ref{thm:eulerD} establishes a topological balance law between the two phases of a tubular tiling, and in the Introduction we presented several geometric examples. However, neither the theorem nor the examples provide an \emph{algorithm} for constructing tubular tilings.

As noted in the Introduction, TPMS such as the Gyroid and the Schwarz $P$ and $D$ surfaces (see Figure \ref{fig:tub}(b1)) provide examples of binary mixtures, the natural discretizations of which are tubular tilings where
the complementary domains $A$ and $B$ can be deformation retracted on two skeletal graphs $G_A,G_B$. Here we follow the opposite path: we assume that the skeletal graph $G_A$ for  phase $A$ (the observable phase
) is given, and we provide an algorithm to construct the corresponding tubular surface $T$ and the tubular tiling $\mathcal T$ which also includes the hidden phase $B$. We describe two examples, both related to  $d=3$ dimensional ambient manifolds $M^d$.

First we consider the case when $M^d=\mathbb R^3$ and the skeletal graph $G_A$ is a polyhedral skeleton. Here we  introduce the \emph{frozen wire algorithm} and illustrate it on the example
of the Fermi surface of Copper. Next we consider the case when $M^d=S^3$ and the skeletal graph $G_A$ degenerates to a discrete set of isolated vertices. Here we introduce the \emph{double bubble  algorithm}
and illustrate it on the thick shell decomposition of the positively curved Robertson-Walker universe model.

\subsection{The Fermi surface of Copper: soft cells on sub-atomic scale}

\subsubsection{Polyhedral skeleta and the frozen wire algorithm}
\begin{figure}
    \centering
    \includegraphics[width=\linewidth]{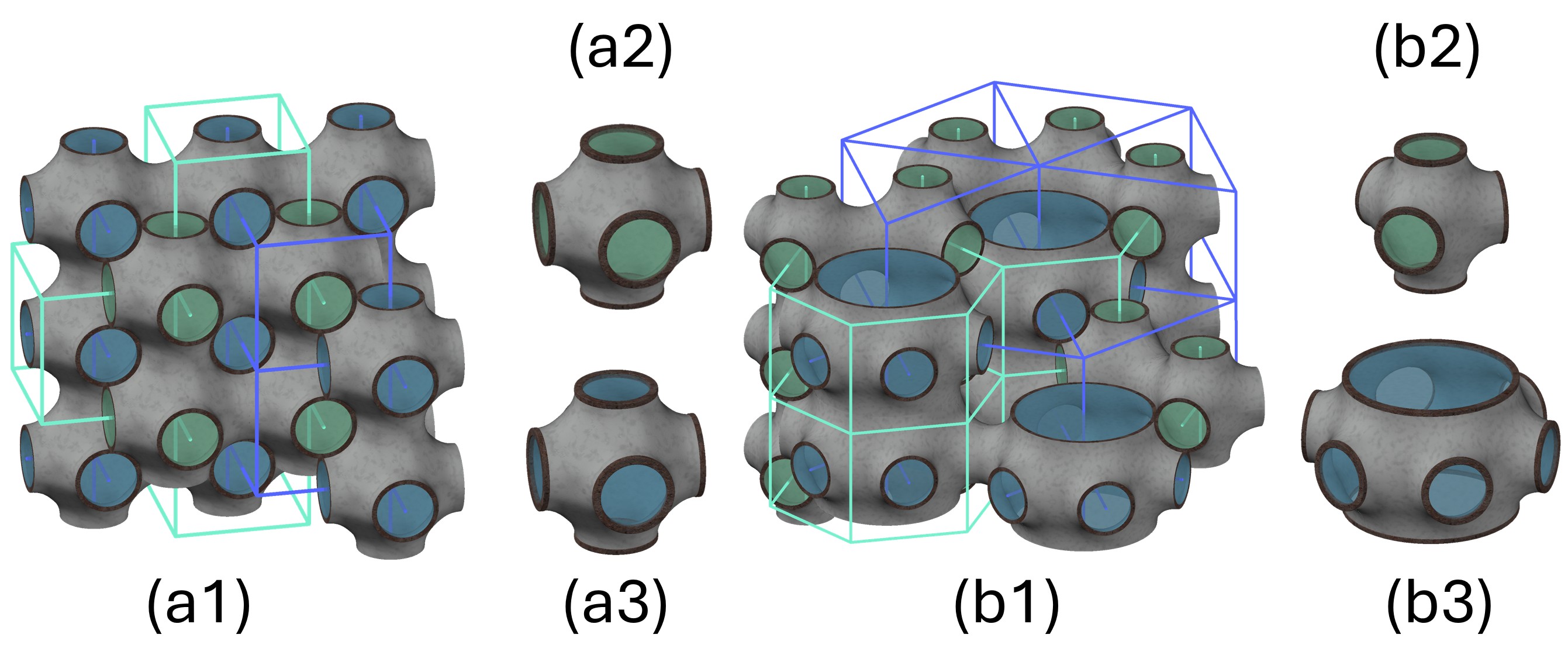}
    \caption{Tubular soft tilings constructed on polyhedral skeletons. (a) The cubic grid: (a1) tiling (a2) $A$-cell (a3) $B$-cell (b) The hexagonal prismatic grid: (b1) tiling (b2) $A$-cell (b3) $B$-cell}
    \label{fig:polyhedral}
\end{figure}
In our first example we consider the case where the skeletal graph $G_A$ of the observable phase is given as a polyhedral skeleton. To establish the tiling for the 
$B$-phase, we first prove a simple corollary concerning polyhedral tilings.

\begin{cor}\label{cor:tub}
Let $\mathcal{M}$ be a normal, balanced convex tiling in $\mathbb R^3$ \cite{Grunbaum}, with $F,V,f,v$ denoting the respective average numbers of faces of tiles, vertices of tiles, faces of vertex polyhedra, and vertices of vertex polyhedra. Then
\begin{equation}\label{eq:thm1}
V(2-v)-f(2-F)=0.
\end{equation}
\end{cor}

\begin{proof}
Let $E(\mathcal M)$ denote the edge skeleton of $\mathcal M$. For sufficiently small tube radius, a tubular neighbourhood of $E(\mathcal M)$ defines a smooth embedded surface $T$ separating $\mathbb R^3$ into two phases. We call this construction the \emph{frozen wire algorithm}.

The resulting tubular tiling has one $A$--cell associated with each vertex of $\mathcal M$  and one $B$--cell inside each polyhedral cell of $\mathcal M$. Since both cell types are simply connected,
\begin{equation}\label{w1}
\chi_{\mathbf A}=\chi_{\mathbf B}=1.
\end{equation}
The internal interfaces of an $A$--cell correspond to the vertices of the vertex polyhedron, while the internal interfaces of a $B$--cell correspond to the faces of the original polyhedron. Hence
\begin{equation}\label{w2}
\chi_{\bar A}=v, \qquad \chi_{\bar B}=F.
\end{equation} 
If $N_C$ and $N_V$ denote the numbers of cells and vertices of $\mathcal M$, respectively, then double counting cell--vertex incidences yields
$N_Vf=N_CV.$ Therefore
\begin{equation}\label{w3}
\frac{p_A}{p_B}=\frac{V}{f}.
\end{equation}
Normality gives a uniform bound on the numbers of faces, edges and vertices of the convex cells and their vertex polyhedra \cite{Grunbaum} and also implies that the relative number of boundary tiles tends to zero. Since the corresponding interface strata in the frozen-wire construction have uniformly bounded Euler characteristic, the associated tubular tiling satisfies the boundary part of topological balance, which is the only part needed in the \(d=3\) homogeneous relation.
Dividing \eqref{eq:thm0D} by $p_B$ and substituting (\ref{w1})-(\ref{w3})  we get
a formula equivalent to (\ref{eq:thm1}).
\end{proof}

\begin{rem}\label{rem:pol}
Let $\mathcal M$ be a polyhedral tiling of $\mathbb R^3$ and let
$E(\mathcal M)$ denote its edge skeleton. We call a graph
$\bar E$ a \emph{superpolyhedral grid} if it has the same nodal set as
$E(\mathcal M)$ and may contain additional straight edges.
If a smooth embedded surface
$T\subset \mathbb R^3$ satisfies
\[
T\cap \bar E=\emptyset,
\]
then $T$ also avoids $E(\mathcal M)$ and therefore defines a tubular surface associated with the tubular tiling obtained from the frozen wire algorithm applied to $E(\mathcal M)$.
\end{rem}

\subsubsection{Geometry of Fermi surfaces}

A crystalline solid determines a lattice in physical space and an associated reciprocal lattice in momentum space. Physically, reciprocal space is considered modulo reciprocal lattice translations, yielding the Brillouin zone, which is topologically a three-dimensional torus. The Fermi surface is the level surface separating occupied and unoccupied electron states.

For the present topological computation it is more convenient to work with the equivalent triply periodic representation in $\mathbb R^3$. Geometrically, the Fermi surface defines a binary decomposition of reciprocal space into two complementary phases separated by a smooth embedded surface. Since only one phase (the occupied phase) is directly observable in many experiments, Fermi surfaces provide natural examples of semi-hidden tubular tilings.

\subsubsection{The Fermi surface of copper}

\begin{figure*}
    \centering
    \includegraphics[width=\linewidth]{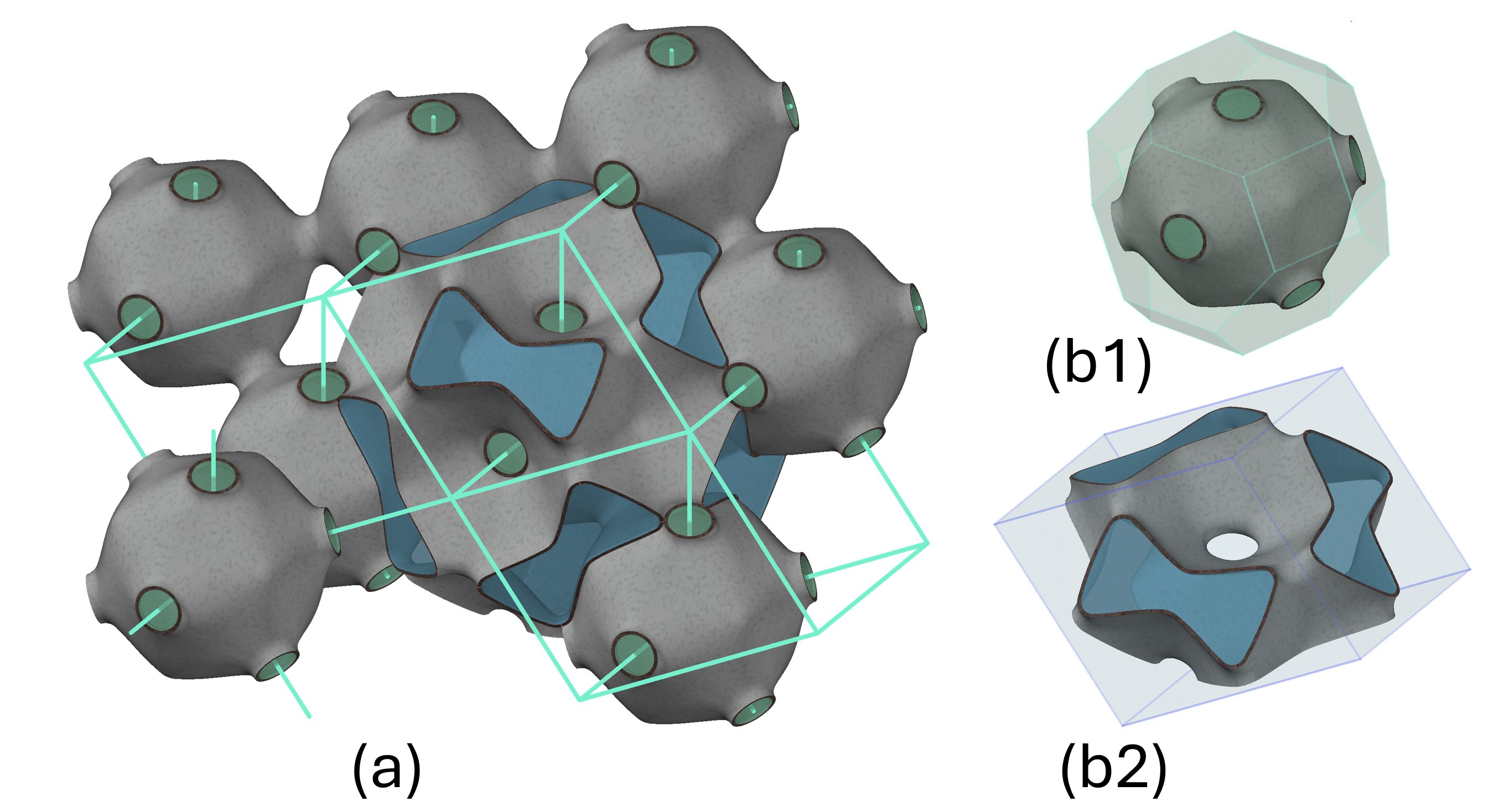}
    \caption{The Fermi surface of copper and the associated tubular tiling. (a) Fermi surface of copper (grey surface) with internal interfaces for the occupied, observable $A$-phase (green disks), internal interfaces for the
		non-occupied, hidden $B$-phase (blue disks) and the $B_D$ superpolyhedral skeleton (green lines). Observe that the $A$-phase avoids the $B_D$ skeleton. (b1) The soft unit cell of the occupied, observable $A$-phase, illustrated inside the truncated octahedron cell. (b2) The soft unit cell with torus topology of the non-occupied, hidden $B$-phase, illustrated inside the rhombohedron.}
    \label{fig:fermi}
\end{figure*}

Let $B_D$ denote the union of the diagonal edges of the BCC lattice. These edges form a superpolyhedral grid since they coincide with the union of the edge skeletons of the two rhombohedral tilings associated with the BCC lattice.

The Brillouin zone of copper is a truncated octahedron. Its Fermi surface is approximately spherical but develops eight necks passing through the centres of the hexagonal faces. Consequently, the Fermi surface avoids the diagonal BCC grid $B_D$ and intersects each hexagonal face in a simply connected patch.

We call the occupied phase the $A$--phase. Since the occupied region inside a truncated octahedron is simply connected,
\begin{equation}\label{eq:cufermi1}
\chi_{\mathbf A}=1.
\end{equation}
The internal interface of an $A$--cell consists of eight disjoint discs, one on each hexagonal face, hence
\begin{equation}\label{eq:cufermi1a}
\chi_{\bar A}=8.
\end{equation}

To construct the complementary $B$--cells we choose one of the rhombohedral tilings contained in $B_D$ and apply the frozen wire algorithm. Since a rhombohedron has six faces, the internal interface of a $B$--cell consists of six disjoint discs, yielding
\begin{equation}\label{eq:cufermi2}
\chi_{\bar B}=6.
\end{equation}

Both the truncated octahedron and the rhombohedron tiles $\mathbb R^3$ without gaps and overlaps and both are associated with the BCC lattice. 
If the edge length of the cube is taken as unit, both cells have volume $1/2$, so they occur with equal asymptotic frequency, so
\begin{equation}\label{eq:cufermi3}
p_A=p_B.
\end{equation}

Since the ambient manifold is $M^d=\mathbb R^3$, Theorem~\ref{thm:eulerD} has vanishing right-hand side. Substituting (\ref{eq:cufermi1})--(\ref{eq:cufermi3}) into (\ref{eq:thm0D}) yields
\begin{equation}\label{eq:cufermi4}
\chi_{\mathbf B}=0.
\end{equation}
This is illustrated in Figure \ref{fig:fermi}(b2) where we can see the $B$-cell with torus topology.
As we can see, the topology of the hidden $B$-phase can be inferred directly from
the topology of the observable phase together with geometric information,
namely the relative volumes of the corresponding cells. Although the Fermi surface determines both complementary regions equally, 
physical intuition naturally singles out the occupied region. The Euler balance theorem restores this symmetry by allowing the 
topology of the complementary region to be inferred from the observed interface.

In the next example we consider a different ambient manifold and a very different
physical setting; nevertheless, the argument follows a similar pattern,
combining topological and geometric information to recover properties of
the hidden phase.

\subsection{The Robertson-Walker universe: soft cells and cosmic shells}\label{ss:robertson}

In the forthcoming example we regard the case where the skeletal graph $G_A$ of the observable phase has no edges and is thus reduced to a finite point set.
First, in Corollary \ref{cor:bubble} we show that in the $M^d=S^3$ ambient space, by using two finite sets of embedded $S^2$ spheres ('double bubbles'), a broad range of soft tilings can be constructed
for this special skeletal graph.
Then we use the double bubble algorithm to construct a soft tiling for the positively curved Robertson-Walker model of the universe where we regard the union of the planets as the observable  $A$-phase and use the corollary to construct the soft tiling of the $B$-phase which corresponds to a thick shell decomposition.

\subsubsection{Degenerate skeleta and the double bubble algorithm}
\begin{cor}\label{cor:bubble}
Let $M^d=S^3$ and  let $F^A_n \subset S^3$
be the union of the boundaries of $n>1$ pairwise disjoint 3-balls, so $F^A_n$ is the union of
$n>1$ disjoint $S^2$ spheres.
 Then, for arbitrary $q <n $ there exist two soft 
tubular tilings, both  with $T=F^A_n$ as tubular surface and $\chi_{\mathbf{A}}=1$, and 
\begin{enumerate}
\item one tiling with $\chi_{\mathbf{B}}>q$ and 
\item another tiling with $\chi_{\mathbf{B}}<q$.
\end{enumerate}
\end{cor}

\begin{proof}
The surface $F^A_n$ decomposes $S^3$ into $n+1$ connected components, $n$ of which are 3-balls
bounded by the $n$ spheres. We 
identify the $A$-tiles with these balls and so we can write:
\begin{equation}\label{eq:Atiles}
\chi_{\mathbf{A}}=1, \quad \chi_{\bar A}=0,
\end{equation}
and for brevity we will refer to the $S^2$ spheres in $F^A_n$ as $A$-spheres.
To construct the tiling of the $B$-phase we introduce $F^B_m$
as a set of $m$ pairwise disjoint, embedded $S^2$ spheres in $S^3$ to which we refer as the $B$ spheres.
We assume that if  the intersection $F^A_n \cap F^B_m$ is not empty then it is transversal and consists of $N$ disjoint copies of $S^1$, to which we refer as $C$-circles.
The transversality assumptions guarantee that the induced decomposition satisfies the tubularity conditions. Indeed, every point of a $C$-circle is incident to exactly two external interface patches and one internal interface patch, so the triple-junction condition is satisfied.

We note that for $m>0$, $N$ may be regarded as an arbitrary nonnegative integer.
Indeed, by introducing sufficiently small local wiggles on the $B$-spheres, one may realize any prescribed finite number of transverse $C$-circles with the fixed family $F_n^A$.
While the  surfaces $F^A_n$ and $F^B_m$ always define a tubular tiling, and their definition is analogous, they play entirely different roles in the
tubular tiling. The union of $A$-spheres is the tubular surface $T$ and each $A$-sphere is the entire external interface
of an $A$-tile and (possibly part of) the external interface of a $B$-tile. On the other hand, a  $B$ sphere 
is (possibly part of) the internal interface of a $B$-tile but it is never part of an external interface.

Let us write the Euler characteristic for the internal interface $\bar \partial B_i$ of the $i$th $B$-tile.
Since the $B$-spheres are disjoint, the interface contains an integer number of spheres which we denote by $b_i$.
We denote the number of $C$-circles on the boundary by $c_i$ and since each boundary circle decreases the Euler characteristic by one, inclusion-exclusion yields
\begin{equation}\label{eq:binternal}
\chi(\bar \partial B_i)=2b_i-c_i, \qquad i=1,\dots ,m+1.
\end{equation}
Our next goal is to rewrite (\ref{eq:binternal}) for the respective averages which we denote by $b$ and $c$. Since there are $m+1$ $B$-tiles, we get
$$b=\frac{1}{m+1}\sum _i b_i, \quad c=\frac{1}{m+1}\sum _i c_i. $$
Each $B$-sphere belongs to the internal interfaces of exactly two
$B$-tiles. Hence
\[
\sum_i b_i=2m.
\]
Likewise, each $C$-circle belongs to the boundaries of exactly two
internal interfaces, so
\[
\sum_i c_i=2N.
\]
Hence
\[
b=\frac{2m}{m+1},\qquad
c=\frac{2N}{m+1}=bk,
\]
where $k=N/m$ is the average number of $C$-circles per $B$-sphere.
Averaging (\ref{eq:binternal}) therefore yields
\begin{equation}
\chi_{\bar B}=2b-c
=\frac{2m}{m+1}(2-k).
\end{equation}
The relative frequencies for the two phases are
\[
p_A=\frac{n}{n+m+1}, \qquad p_B=\frac{m+1}{n+m+1}.
\]
Now we can substitute into (\ref{eq:thm0D}). Since $\chi(S^3)=0$, the right hand side is zero and we obtain:
\begin{equation}\label{eq:thm0D1}
\chi _{\mathbf{B}}=\frac{2m-N+n}{m+1}.
\end{equation}
For $m=N=0$ we obtain $\chi _{\mathbf{B}}=n$, proving the existence of the first tiling in Corollary \ref{cor:bubble}.
Next we fix $n$ and $m$. Then, since $N$ may be arbitrarily large, the right-hand side of (\ref{eq:thm0D1}) can be made smaller than $q$, proving the second claim.
\end{proof}

\begin{figure}
    \centering
    \includegraphics[width=\linewidth]{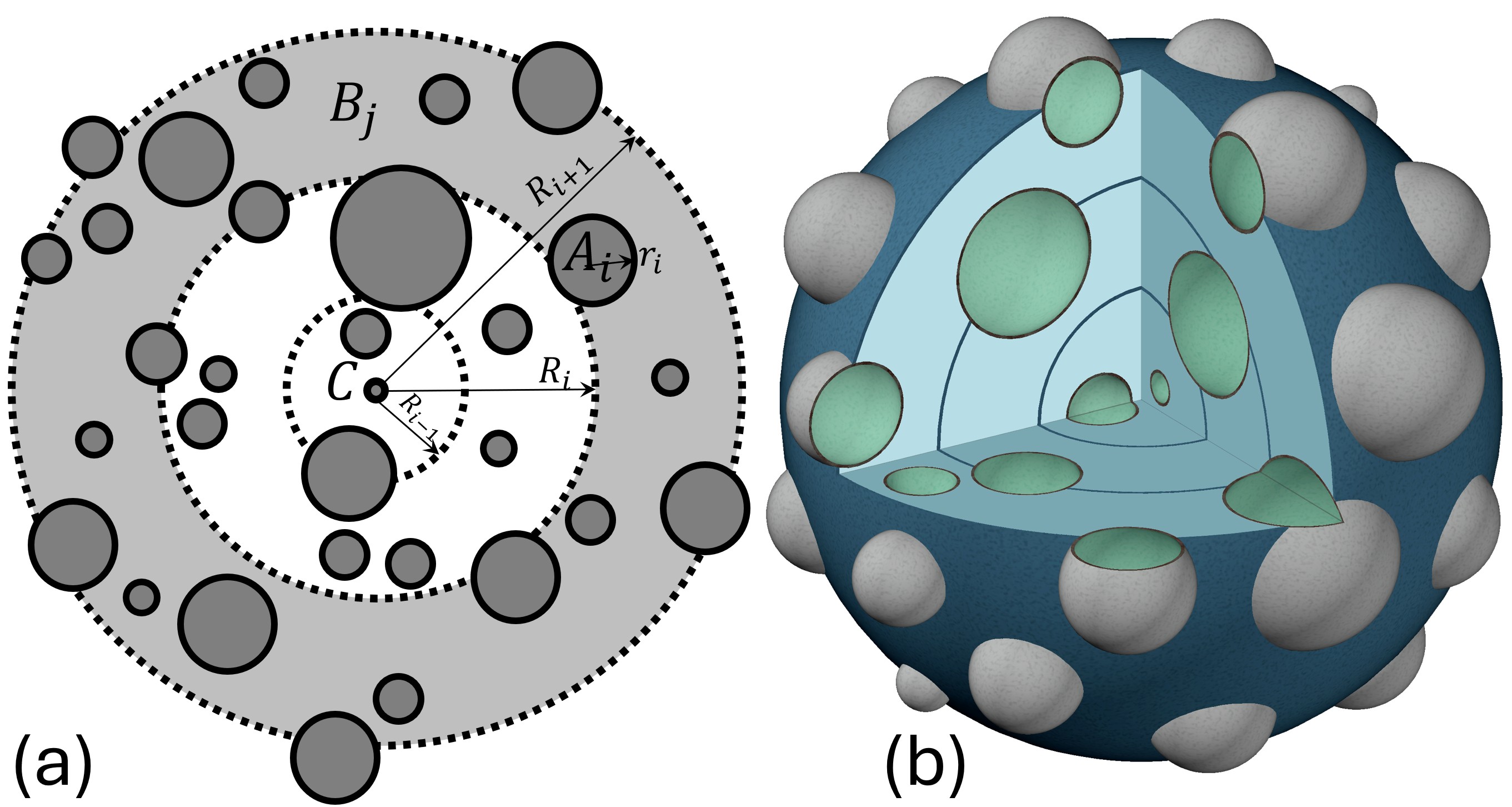}
    \caption{Thick shell decomposition of the $S^3$ Robertson-Walker universe interpreted as a soft tubular tiling. (a) Schematic cross section. Observer located at center $C$. Solid line (union of circles with radii $r_i$, $i=1,2,\dots n$, $n$ being the number of planets) corresponds to smooth tubular surface $T$ which is the union of $n$ disjoint copies of $S^2$ spheres. Domain $A$ is the interior of planets, filled with grey. Each planet is one soft tile $A_i$,$i=1,2, \dots n$. Dotted line (union of circles with radii $R_j=jR_0$, $j=1,2, \dots m$) corresponds to the internal interfaces between $B$-tiles, each of which is a spherical shell with thickness $R_0$. One thick shell filled with light grey. (b) 3D visualization of the double bubble construction for the $S^3$ RW universe. Tubular surface (external interface) appears as a collection of small grey spheres. Internal interface of $B$ tiles appears as large, concentric blue spheres. The $A$-phase is the interior of $A$-tiles (planets), shown in green.  Observe that despite transversal intersection, all boundary strata are smooth and no sharp corners arise, hence the tiling is soft.}
    \label{fig:robertson}
\end{figure}

\subsubsection{The RW universe and its radial shell decomposition}
The Robertson--Walker (RW) model of the Universe
\cite{Robertson1935,Walker1937} assumes spatial slices of constant
curvature. Depending on the sign of the curvature, the corresponding
spatial geometry is modeled on $S^3$, $\mathbb R^3$ (or compact
quotients such as $T^3$), or $H^3$.
Here we discuss the positively curved $S^3$ case.
We assume that the  $S^3$ RW universe is decomposed into a finite set of concentric thick spherical
shells with center $C$. This radial shell decomposition underlies many standard constructions
  in cosmology, such as radial binning in galaxy surveys
  \cite{Peebles1980LargeScaleStructure,Davis1983GalaxyCorrelation},
  tomographic analysis of large–scale structure
  \cite{Kitching2014WeakLensingTomography}, and harmonic
  decompositions based on spherical Fourier–Bessel basis functions
  \cite{Heavens1995ThreeDimensionCosmology,Fisher1995SFB}. It provides
  a natural way to discretize the Universe into thick shells of fixed
  radial width around an observer.
	
	In the RW $S^3$ universe a finite set of thick shells exhaust the entire domain
and, due to the simple internal metric,  the radial distance corresponds to geodesic distance.
If we consider the $S^3$ universe embedded into $\mathbb R^4$ then the thick shell decomposition with center $C$ 
is equivalent to slices of $S^3$ by hyperplanes which are parallel to the tangent hyperplane at $C$.
Next we show that the radial shell decomposition of the RW Universe can be identified with a semi-hidden tubular soft tiling.

We emphasize that the purpose of this example is not to derive a new
cosmological result. Rather, the Robertson--Walker setting provides a
geometric framework, inspired by standard radial shell
decompositions used in cosmology, in which a semi-hidden tubular
tiling arises naturally. Our aim is to illustrate how
Theorem~\ref{thm:eulerD}  organizes the topology of such a tiling and
permits the inference of less obvious topological quantities associated
with the hidden phase.

\subsubsection{Interpretation of the double bubble formula to the RW universe}
The double bubble algorithm of Corollary \ref{cor:bubble} can be directly applied to this
system if we assume that the planets are modeled by 3-balls and their boundaries
are the $A$-spheres.  The concentric shells of the RW decomposition correspond to the $B$-spheres.
We can re-write (\ref{eq:thm0D1}) as:
\begin{equation}\label{eq:RW2}
\chi_{\mathbf B} =\left(2+\frac{1}{m}(n-N)\right)\frac{m}{m+1}.
\end{equation}
As we can observe, the topology of the $B$-cells is essentially controlled by the term $(n-N)/m$
which we can interpret heuristically.
If there would be no planets $(n=0)$ and no intersections $(N=0)$ then, except for the two
3-balls $B_1, B_m$ at the north and south poles,  all remaining $B$-cells would be spherical shells
with  $\chi(B_i)=2$, $(i=2,3, \dots m-1)$. Planets can cause two types of local "topological defects":
they can either appear as enclosed cavities (adding 1 to the Euler characteristic) or they can intersect both
boundary shells, resulting in a hole (subtracting 1 from the Euler characteristic). 	If $n=N$ then these
two effects cancel each other because, on average, each planet ($A$-sphere) is intersected by exactly one $B$-sphere,
producing one $C$-circle, so the intersection produces neither a cavity nor a hole. If $n>N$ then cavities dominate and if $n<N$ then
holes dominate. 

Corollary \ref{cor:bubble} showed that, in the general case, the range of $\chi_{\mathbf B}$ is infinite. However, if we
also consider metric information on the topological $A$-spheres and $B$ spheres then this range becomes much more 
restricted and may even be narrowed down to a single value.
In the case of the RW universe we treat a very special case of the double bubble algorithm because both
the $A$-spheres and the $B$-spheres are not just topological spheres but also spheres in the metric sense.
Since one $A$-sphere and one $B$-sphere can intersect transversely 
in at most one $C$-circle, so we have $N \leq nm$.
This constrains the Euler characteristic $\chi_{\mathbf B}$ of the hidden $B$-tiles 
more than the condition in Corollary \ref{cor:bubble} and we obtain the finite bounds
\[
2-n < \chi_{\mathbf B}\leq n,
\]
with the lower bound corresponding to the $m \to \infty$ limit, the upper bound attained at $m=0$. 
As we can observe, the admissible range of the topological quantity $\chi_{\mathbf{B}}$
can be narrowed using  metric information.

By adding further metric information, for fixed $n,m$ the actual expected number $N$ of intersections can be also  expressed 
and  the admissible interval for $\chi_{\mathbf B}$ can be narrowed to a single value. Let us assume that centers of planets are uniformly distributed in space 
and let us denote the average radius of a planet by $r$. Let us further assume that spacing between $B$-spheres is uniform
at radial distance $R_0$  (see Figure \ref{fig:robertson}). Then, geometric probability suggests that the average number $N/n$ of $C$-circles per planet can be written as
$N/n \approx 2r/R_0$. Fixing $n,m$ and substituting $N=2rn/R_0$ into (\ref{eq:RW2}) yields $\chi_{\mathbf B}$.

Both examples discussed in this section exhibit an observable phase and a  hidden phase.
In both cases the topology of the soft tiling for the hidden phase could be determined by relying on two sources of information:
  equation (\ref{eq:thm0D}) in Theorem \ref{thm:eulerD} and additional, geometric information. In the Fermi-surface example the latter
	was derived from the geometry of the BCC lattice and, jointly with equation (\ref{eq:thm0D}) this determined the hidden phase essentially uniquely. In the RW example metric spherical geometry reduced 
	an initially infinite range of possibilities to a finite interval and using geometric averaging the actual value of $\chi_{\mathbf B}$ could be computed.


\section{Summary}

In this paper we introduced \emph{tubular tilings}, natural discretizations of binary mixtures defined on smooth manifolds by separating interfaces. We proved that every tubular tiling in dimensions $d>2$ has no boundary  strata of codimension greater than two. We also derived a global Euler balance law (Theorem~\ref{thm:eulerD}) relating the Euler characteristics of tubular tiles, their internal interfaces and the ambient manifold.

The definition of tubular tilings used in this paper involves two
distinct structural requirements. First, the separating hypersurface
$T$ is required to be smooth. Second, the binary labeling of the
tiling together with the associated interface is required to satisfy
the \emph{tubularity condition}, formulated as a triple-junction criterion
in subection \ref{ss:tt}, equation (\ref{eq:triple}).

More generally, one may formulate tubularity conditions by a pair $(n,k)$,
requiring that the common intersection of any $n$ distinct tiles of the
same phase have dimension strictly smaller than $k$. The tubularity
condition used here is the $(3,0)$ condition: any three distinct tiles
of the same phase have a common intersection of dimension strictly
smaller than zero, and hence have empty intersection. Together,
smoothness of the interface and $(3,0)$ tubularity imply softness of
the tiling; this property is proved in Proposition \ref{prop:soft}
(subsection \ref{ss:soft}). For piecewise-smooth interfaces associated
with polyhedral tilings and foams, conditions with $k>0$ may provide
natural generalizations, while $n>3$ may lead to higher-order
topological balance laws.

The notion of $k$-softness also suggests finer geometric classifications. While $k$-softness excludes boundary strata of codimension greater than $k$, one may further distinguish $k$-soft tilings by suitable integral measures supported on codimension-$k$ strata. In particular, for 2-soft monohedral tilings of $\mathbb{R}^3$, where the lowest-dimensional boundary strata are edges, a natural candidate is the total Gaussian curvature integrated along all edges. We conjecture that this quantity satisfies the sharp lower bound $K_E\ge 6\pi$. If true, the tubular soft cell associated with the Gyroid (Figure~\ref{fig:0}(f)) would realize the extremal configuration. More broadly, such extremal principles would complement the topological balance laws developed here by linking them to quantitative geometric optimization.

\begin{figure*}[ht]
    \centering
    \includegraphics[width=0.95\linewidth]{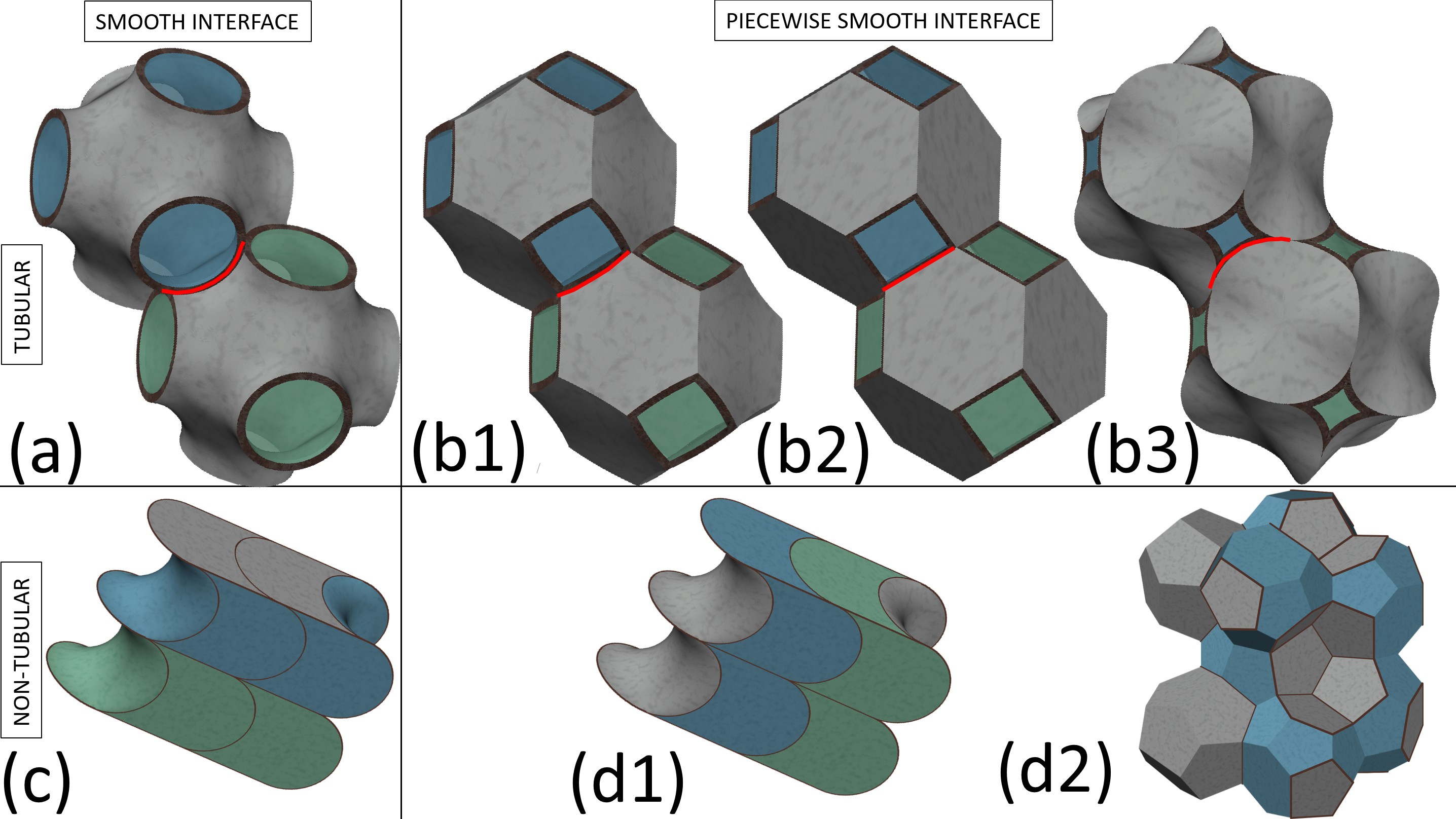}
    \caption{Four logical classes of interface-based cellular structures,
classified according to the smoothness of the separating interface
(columns) and the tubularity criterion (rows). The upper row contains
four realizations with the same combinatorial structure, belonging to
a one-parameter family \cite{softcells2} parametrized by the curvature
of marked interface curves (red). (a) Smooth interface: the tubular
tiling associated with the Schwarz $P$ surface. (b) Piecewise-smooth
interfaces: (b1) Kelvin foam; (b2) the truncated-octahedral tiling
with the $\{6,4|4\}$ skew polyhedron as interface
\cite{Coxeter,Schoen1970}; and (b3) the $f2$ soft tiling introduced in
\cite{softcells1}. Bottom row: tubularity criterion violated.
(c) The $d3$ soft tiling from \cite{softcells1}, obtained from the
cubic grid by the Edge Bending Algorithm; its separating interface is
a collection of parallel planes. (d1) The same underlying tiling with
a different choice of separating interface, violating both
smoothness and tubularity. (d2) A polyhedral approximation of the
Weaire--Phelan structure, with phases $A$ and $B$ corresponding to
the two types of polyhedral cells.}
    \label{fig:discretized}
\end{figure*}

The smoothness and the tubularity requirements have different roles. The smoothness of $T$
is motivated by the interpretation of $T$ as an idealized interface
between two phases. On the other hand,  some piecewise-smooth interfaces satisfying the same $(3,0)$ tubularity criterion provide natural discretized analogues.
While simultaneously applied the two requirements enforce softness, below we show that neither requirement on its own is
necessary for softness.

Relaxing the tubularity and smoothness conditions independently gives four logically distinct classes of interface-based
cellular structures. Figure~\ref{fig:discretized} illustrates these
four classes by representative examples.

The example in Figure~\ref{fig:discretized}(a) is the tubular tiling associated with the
Schwarz $P$ surface: the separating interface is smooth and the
tubularity criterion is satisfied. The upper-right entry (Figure~\ref{fig:discretized}(b))
contains three piecewise-smooth realizations which are combinatorially
equivalent to the soft tiling of the Schwarz P surface \cite{softcells2} and the tubularity criterion
is also satisfied: 
\begin{enumerate}
\item The Kelvin foam, based on the truncated octahedron, is a classical model of an ordered foam \cite{Weaire_Hutzler}.
\item The truncated-octahedral polyhedral tiling with the $\{6,4|4\}$ skew polyhedron \cite{Coxeter, Schoen1970} as interface.
\item The $f2$ soft tiling introduced in \cite{softcells1}, see also Figure \ref{fig:0}(a).
\end{enumerate}
 As pointed out in \cite{softcells2}, together with the Schwarz $P$ realization,
these structures belong to a one-parameter family which can be suitably parametrized
by the curvature of a marked interface curve (see red curves in Figure \ref{fig:discretized}(b)). This one-parameter family illustrates
how a smooth interface may be replaced by a piecewise-smooth or
polyhedral one while retaining the combinatorial structure expressed
by the tubularity criterion.

The lower-left example (Figure~\ref{fig:discretized}(c)) shows the $d3$ soft tiling,
derived from the cubic grid via the Edge Bending Algorithm \cite{softcells1}, see also Figure \ref{fig:0}(c).
Here, similar to the brick wall in Figure \ref{fig:tub}, the interface is an infinite
set of parallel hyperplanes, illustrating that a tiling  may have a smooth
separating interface even when the tubularity criterion is violated.
The first lower-right example (Figure~\ref{fig:discretized}(d1)) has the same underlying tiling as the
previous one, however, the separating interface is selected differently, so that the binary
structure violates both requirements while nevertheless
giving a soft tiling.

A further non-soft example in this category is the Weaire--Phelan
structure (Figure~\ref{fig:discretized}(d2)), with the phases $A$ and
$B$ identified with the two types of polyhedral cells. For this
coloring, the resulting binary interface violates the $(3,0)$
tubularity criterion because some edges are incident to three
polyhedra of the 14-faced type, giving $AAA$ edges. This failure appears to be
not specific to the conventional coloring, however. Explicit
computation using the periodic cell complex provided by Brakke
\cite{BrakkeWP} indicates that no binary coloring of the Weaire--Phelan
structure satisfies either the $(3,0)$ or the $(3,1)$ tubularity
condition. The $(3,2)$ condition, in contrast, is automatically
satisfied by any binary coloring of a face-to-face polyhedral
mosaic.

This perspective places tubular tilings in a broader family of
cellular structures that includes polyhedral tilings, soft cellular
structures, and foams. In particular, polyhedral cellular structures
may provide discrete realizations of interface structures analogous
to the smooth interfaces considered here. Foams can likewise be
equipped with binary colorings whose associated interfaces have
nontrivial topological structure. Whether a particular binary
coloring has a direct physical interpretation depends, however, on
the underlying physical system and is not implied by the geometric
construction alone.

\section{Acknowledgments}
The author is indebted to Gerg\H o Alm\'adi who
geometrically constructed the toroidal $B$-cell of the Fermi Copper surface, to Viola Radics who constructed the tubular tiling for the hexagonal prismatic grid and
to L\'aszl\'o Strommer who designed all figures. 
The author thanks \'Akos G. Horv\'ath  and the anonymous reviewers for their valuable suggestions. 

\appendix
\section{Basic concepts}\label{sec:basic}

\subsection{Tubular tilings}\label{ss:tt}

\begin{defn}[Tubular tilings]\label{defn:tubulartilings}
Let $M^d$ be a connected, complete Riemannian $d$--manifold without boundary, and let $\mathcal{T}_M$ be a locally finite tiling of $M^d$ with connected, compact tiles $\tau_i$.  
We require that there exist constants $0<r_-\le r_+<\infty     $ such that every tile $\tau_i$ contains an intrinsic ball of radius $r_-$ and is contained in an intrinsic ball of radius $r_+$.

We also require that the interface components
\[
\tau_{i,j} = \tau_i \cap \tau_j
\]
are either empty or smooth, compact $(d-1)$--dimensional manifolds
and we assume that the boundary of each tile is given by pairwise intersections with neighboring tiles:
\[
\partial \tau_i = \bigcup_{j \neq i} (\tau_i \cap \tau_j).
\]
We also assume that any two distinct connected components $\tau _{i,j}$ of interface sets $\tau_i \cap \tau_j$ intersect transversely whenever they intersect.

Each tile is assigned a label of the form $A_i$ or $B_j$, where indices distinguish tiles within each class, i.e. $A_i \neq A_j$ for $i \neq j$, and similarly $B_i \neq B_j$. Thus we obtain a \emph{binary labeling} $\mathcal{T}^{(A,B)}_M$ of the tiling. 
Each tile $\tau_i$ is assigned exactly one label $A_j$ or $B_k$, and we identify the tile with its label for notational convenience.
We define the structure for the boundary $\partial A_i$ for a tile $A_i$. We call
\begin{equation}\label{eq:extinterface}
\hat \partial A_i := \bigcup_j A_i \cap B_j 
\end{equation}
its \emph{external interface} and
\begin{equation}\label{eq:intinterface}
\bar \partial A_i := \bigcup_{j\neq i} A_i \cap A_j 
\end{equation}
its \emph{internal interface}.  Using these concepts, the boundary of the tile can be decomposed as
\begin{equation}\label{eq:cellboundary}
\partial A_i =\hat{\partial}A_i\cup\bar{\partial}A_i, 
\end{equation}
and we call
\begin{equation}\label{eq:sepboundary}
\mathring{\partial}A_i
:=\hat{\partial}A_i\cap\bar{\partial}A_i,
\end{equation}
the \emph{separation boundary} (see Figure \ref{fig:boundary}).
We call $\mathcal{T}_M$ a \emph{tubular tiling} if there exists a smooth embedded hypersurface $T \subset M^d$ and a binary labeling $\mathcal{T}^{(A,B)}_M$ such that:
\begin{equation}\label{eq:ttling}
\bigcup_{i,j} A_i \cap B_j  =   T, 
\end{equation}
and \begin{equation}\label{eq:triple}
A_i \cap A_j \cap A_k = B_i \cap B_j \cap B_k = \emptyset \qquad \text{for all pairwise distinct} i,j,k.
\end{equation}
If these conditions are met then we call $T$ a \emph{tubular (hyper)surface}.
\end{defn}

\begin{figure}
    \centering
    \includegraphics[width=0.25\linewidth]{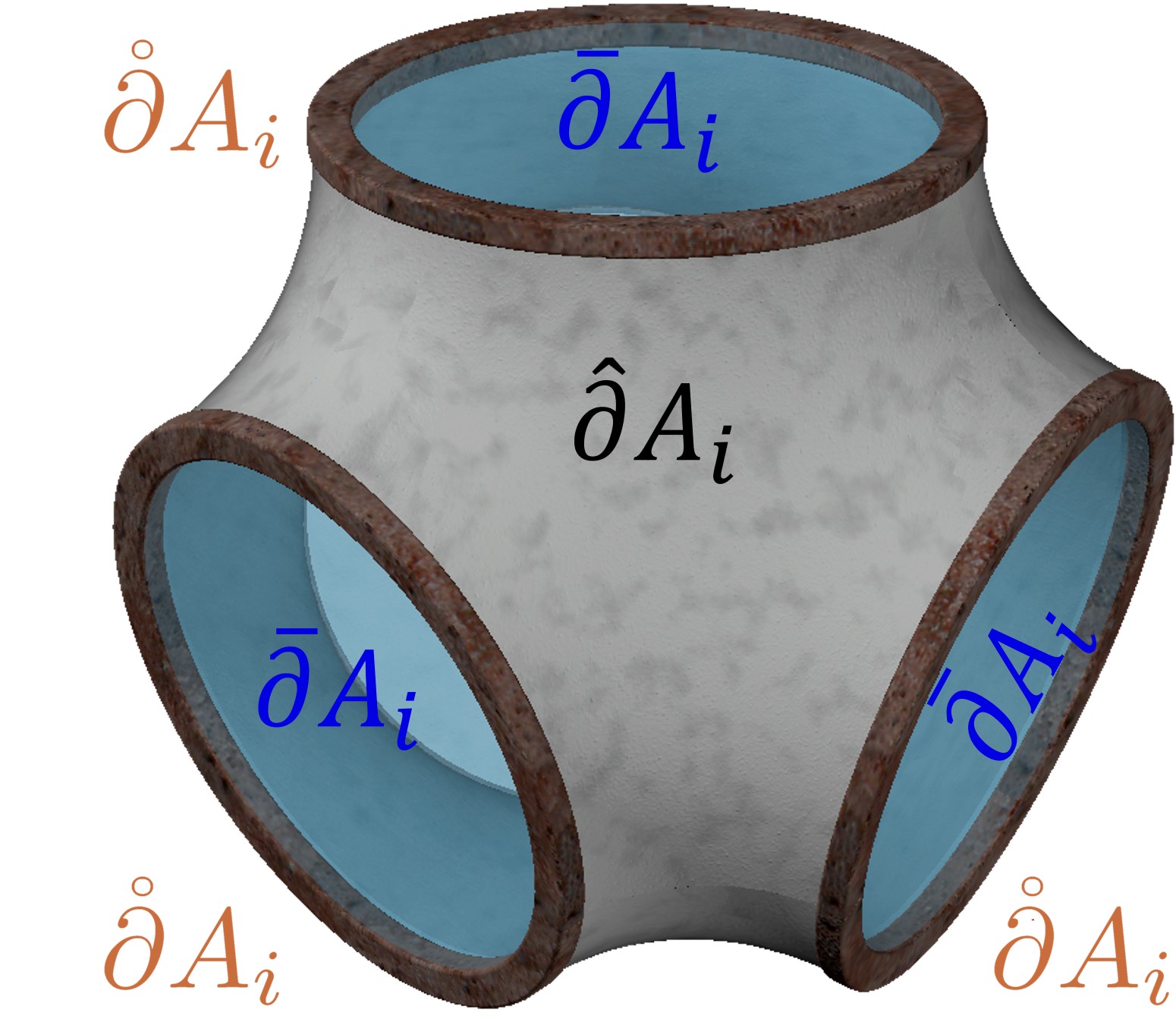}
    \caption{Boundary structure of tubular cell.}
    \label{fig:boundary}
\end{figure}

\begin{rem}[Boundary decomposition of a tubular tile, perfect tilings]
For the boundary $\partial A_i$ of each tile $A_i$, equations~(\ref{eq:cellboundary}) and~(\ref{eq:sepboundary}) define a decomposition into the external and internal interfaces $\hat{\partial}A_i$ and $\bar{\partial}A_i$. Both interfaces are unions of compact $(d-1)$--manifolds, intersecting transversely along smooth compact $(d-2)$--manifolds. The union of these intersections is called the \emph{separation boundary} and is denoted by $\mathring{\partial}A_i$ (see Figure~\ref{fig:boundary}).

We say that a tile is \emph{active} if its external interface is nonempty, and \emph{passive} otherwise. In the latter case the separation boundary is empty. 
A tubular tiling is called \emph{perfect} if all of its tiles are active.
Theorem 1 remains valid for imperfect tilings satisfying its hypotheses.
Nevertheless, many physical systems appear naturally as perfect tubular tilings.
 For an imperfect tubular tiling, one may often obtain a perfect tiling either by amalgamating passive tiles with neighboring tiles of the same label or by relabelling the tiling. The first operation preserves the tubular hypersurface $T$ while modifying the combinatorial structure of the tiling, whereas the second preserves the combinatorial structure while modifying the associated tubular hypersurface.
\end{rem}

\begin{rem}[Induced decompositions of the interface hypersurface]
Let $\mathcal{T}_M$ be a tubular tiling with associated hypersurface
$T\subset M^d$. Then the external interfaces
\[
A_i\cap B_j\subset T
\]
define two natural decompositions of $T$.

From the $A$--side we have
\[
T=\bigcup_i \hat\partial A_i
   =\bigcup_i\Bigl(\bigcup_j A_i\cap B_j\Bigr),
\]
and from the $B$--side
\[
T=\bigcup_j \hat\partial B_j
   =\bigcup_j\Bigl(\bigcup_i A_i\cap B_j\Bigr).
\]

If the tubular tiling is perfect, then every tile has nonempty
external interface and these decompositions may be viewed as two
natural tilings of $T$, corresponding to the interface structure seen
from the $A$--phase and the $B$-phase, respectively.
In particular, codimension-$2$ strata of the tubular tiling in $M^d$
appear as codimension-$1$ interfaces in the induced decompositions of
$T$.
\end{rem}

\begin{defn}\label{defn:topological balance} (Topological balance).
Let $M^d$ be a connected, complete Riemannian $d$--manifold without boundary, and let $\mathcal{T}_M$ be a tubular tiling of $M^d$ with  tiles $\tau_i$.  
Fix \(P\in M^d\). For \(R>0\), let \(\mathcal B_R(P)\) denote the closed geodesic ball of radius \(R\) centered at \(P\). Let $N(R)$ be the number of tiles entirely contained in ${\mathcal{B}}_R(P)$, let $\mathcal D_R(P)$ be the set of tiles intersecting its boundary $\partial{\mathcal{B}}_R(P)$ and let $n(R)=\#\mathcal D_R(P)$.
We call $\mathcal T_M$ \emph{topologically balanced at $P$} if
\[
\lim_{R \to \infty} \frac{n(R)}{N(R)}= 0,
\]
and there exist $C_P<\infty$ and $R_P>0$ such that, for every
$R\geq R_P$ and every tile $\tau_i\in\mathcal D_R(P)$, the sum of
the absolute Euler characteristics of all its interface components
$\tau_{i,j}$ and of all nonempty pairwise intersections among these components
is at most $C_P$. We also require that the Euler characteristics of the unions
of all tiles entirely contained in ${\mathcal{B}}_R(P)$ remain bounded in
absolute value for sufficiently large $R$.
We call $\mathcal T_M$ \emph{topologically balanced} if it is so at
every $P\in M^d$. If $M^d$ is compact, the condition is automatic for
sufficiently large $R$.
\end{defn}

We remark that topological balance is distinct from the standard notion of a balanced tiling \cite{Grunbaum}; neither is intended as a special case of the other. The two notions are nevertheless related through the idea of large-scale averaging.

\subsection{Soft shapes and soft cells}\label{ss:soft}
Softness of shapes and tilings was first defined in \cite{softcells1} for $d=2$ and $d=3$ dimensions by prescribing that the number of sharp corners should be minimized.
Here we aim, on one hand, to extend this notion to arbitrary dimensions. On the other hand, we want to connect the concept of softness and the concept of smoothness.
The following definition of softness $\sigma$ measures what is the boundary stratum with highest codimension. If $\sigma=1$ then the boundary is smooth. In $d=3$ dimensions, $\sigma=2$
characterizes the same shapes identified by the definition in \cite{softcells1}. In $d=2$ dimensions the new definition is vacuous and therefore we suggest to keep the old definition
because it selects shapes also observed in nature.

\begin{defn}[Soft shape]\label{defn:soft1}
Let $M^d$ be a smooth $d$--dimensional manifold without boundary
 and let $S \subset M^d$ be a compact, connected, orientable $d$--manifold with piecewise smooth boundary $\partial S$.  

We call $S$ a \emph{$k$-soft shape} (or, alternatively we write $\sigma(S)=k$) if $k$ is the smallest integer such that for every point $p \in \partial S$ there exists a smooth embedded $k$-codimensional submanifold
\[
\gamma \subseteq \partial S
\]
such that $p$ lies in the relative interior of $\gamma$.
\end{defn}

\begin{rem}[Soft cell terminology]\label{rem:softterm}
Regarding the parameter $k$ we make some simple observations:
\begin{itemize}
\item Since $\dim(\partial S)=d-1$, so $k>0$.
\item If $k=1$  then every boundary point $p$ admits a neighborhood in $\partial S$ that is a smooth $(d-1)$-manifold. Hence $\partial S$ is smooth.  In other words, we recognize that being $1$-soft is equivalent to being smooth.
\item For brevity, if $d>2$ then we call $2$-soft shapes simply soft shapes.
\end{itemize}
\end{rem}

\begin{rem}\label{rem:2Dsoft}
For $d=2$, every boundary point of a piecewise smooth curve lies in the interior of a smooth 0-manifold (a point). Hence every planar shape is automatically $2$-soft, so the definition carries no geometric information. Therefore, in dimension two we retain the earlier notion of softness introduced in \cite{softcells1}.
\end{rem}

\begin{defn}[Soft tiling and soft cell]\label{defn:soft2}
Let $M^d$  be a smooth $d$--dimensional manifold without boundary
and let $0 < k < d$ be an integer.  
If a tiling $\mathcal{T}_M$ consists of tiles, each of which is a $k$-soft shape, we call $\mathcal{T}_M$ a \emph{$k$-soft tiling}.  
If copies of a single $k$-soft shape tile $M^d$, we call that shape a \emph{$k$-soft cell}.
\end{defn}

\begin{prop}\label{prop:soft}
Let $M^d$  be a smooth $d$--dimensional manifold without boundary, and let 
$\mathcal{T}_M$ be a tubular tiling of $M^d$. Then, if $d>2$, the tiling $\mathcal{T}_M$ is a $2$-soft tiling.
\end{prop}

\begin{proof}
Let $A_i$ be a tile and let $p \in \partial A_i$.  
Because the tiling is tubular, every interface component is a smooth $(d-1)$--manifold. Moreover, the condition (\ref{eq:triple})
implies that at $n\leq 3$ interface components (faces) of $\partial A_i$ can meet at the point $p$: since two internal faces already force an $AAA$-violation 
of (\ref{eq:triple}), regardless of any other faces present, and three external faces force a $BBB$-violation, any configuration with $n\ge4$ automatically contains one of these forbidden patterns.
As we can see, external multiplicity never creates a genuine crease, so only internal-face count and the internal/external mix matter. 
If we have $n\leq 3$ then only the following combinations are possible:
\begin{enumerate}
\item $n=1$:
\begin{enumerate}
\item  $p$ is an interior point of one internal face  of $A_i$. Then $p$ lies on a smooth $(d-1)$--manifold contained in $\partial A_i$, thus $A_i$ is $1$--soft at $p$.
This also implies that $A_i$ is $2$--soft at $p$.
\item  $p$ is an interior point of one external face  of $A_i$. Argument is the same as above.
\end{enumerate}
\item $n=2$: 
\begin{enumerate}
\item $p$ is on the boundary between an external and an internal face of $A_i$. Let $f_i^1 \subset T$ be the external face and $f_i^2 = A_i \cap A_k$ be the internal face.
Both $f_i^1$ and $f_i^2$ are smooth $(d-1)$--manifolds. Their intersection is
\[
f_i^{1,2} = f_i^1 \cap f_i^2 = T \cap (A_i \cap A_k),
\]
which is a smooth $(d-2)$--manifold because it is the transverse intersection of two smooth hypersurfaces in $M^d$.
Thus $p$ lies in the interior of the smooth $(d-2)$--manifold $f_i^{1,2}$, and $A_i$ is $2$--soft at $p$.
\item $p$ is on the boundary between two internal faces.  This case can be excluded because it would violate the triple-junction criterion (\ref{eq:triple}).
\item $p$ is on the boundary between two external faces. Since both external faces are patches of the smooth tubular surface $T$, this case reduces to case (1)(b).
\end{enumerate}
\item $n=3$
\begin{enumerate}
\item $p$ is at the intersection of three internal faces of $A_i$. This case can be excluded because it would violate the triple-junction criterion (\ref{eq:triple})
\item $p$ is at the intersection of three external faces of $A_i$.This case can be excluded because it would violate the triple-junction criterion (\ref{eq:triple})
\item $p$ is at the intersection of one external face and two internal faces of $A_i$. This case can be excluded because it would violate the triple-junction criterion (\ref{eq:triple})
\item $p$ is at the intersection of one internal face and two external faces of $A_i$. Since both external faces are patches of the smooth tubular surface $T$, this case reduces to case (2)(a).
\end{enumerate}
\end{enumerate}

In all cases, each boundary point $p \in \partial A_i$ lies in the interior of a smooth $(d-2)$--dimensional submanifold of $\partial A_i$. Hence $A_i$ is a $2$--soft shape.  
Since this holds for each tile, $\mathcal{T}_M$ is a $2$--soft tiling.
\end{proof}

\begin{rem}
The restriction $d>2$ is imposed not because the statement fails for $d=2$, 
but because in that case the notion of $2$--softness becomes vacuous: 
a smooth $0$--manifold is just a point, so every boundary point already satisfies the condition, 
and the proposition carries no geometric content.
If we regard points as smooth manifolds then all polygonal tilings
will be classified as soft tilings and for this reason, the case $d=2$ is excluded from the present definition and
will be treated separately. A slightly different approach, motivated by natural examples,
aims to minimize the number of geometric corners. In $d\geq 3$ dimensions any soft tiling satisfies this
condition because the number of corners is zero. For the 2D case, in \cite{twovertex} it was proven that 
in a normal, balanced  tiling of $\mathbb R^2$ the minimal number
for the average number of geometric corners is 2. Planar tilings with this property are also called soft.
\end{rem}

\begin{rem}
Tubular soft tiles $A_i$ satisfy a strengthened softness property that is not
required by Definition~\ref{defn:soft1}.
Not only is the boundary $\partial A_i$ a closed, piecewise smooth $(d-1)$--manifold,
but the $(d-2)$--dimensional separation boundary
$\mathring{\partial} A_i$
is a \emph{disjoint union} of closed, smooth manifolds.
In order to satisfy Definition~\ref{defn:soft1}$,$ this is not necessary:
the smooth $(d-2)$--manifolds defining the separation boundary could, in principle,
have transverse or tangential intersections. In \cite{twovertex} and  \cite{softcells1} such soft cells with non-disjoint
separation boundaries were presented.
\end{rem}

\begin{rem}
 A \emph{tubular soft cell} is a $d$-dimensional compact object $S$ such that
identical copies of $S$ fill $M^d$  without gaps and overlaps and the boundary $\partial S$ consists of $n$ disjoint, $d-1$ domains with smooth, $(d-2)$ dimensional boundaries and 
$m$ disjoint, smooth, $(d-1)$-dimensional patches of a smooth manifold. It is important to recognize that these properties are sufficient to guarantee that $S$ is a soft cell but they are not sufficient
to guarantee that $S$ is a tubular soft cell. For the latter, the global existence of a smooth, $(d-1)$ dimensional separating surface is also required.
The soft cell illustrated in Figure 4(b) of the main article has all the listed properties but does not belong to a tubular tiling.
\end{rem}

\section{Proof of the Euler balance theorem (Theorem 1)}\label{sec:thm_euler}
\begin{proof} (Theorem 1)
We first introduce the finite volume setup for noncompact manifolds
then write general equations (independent of the parity of the dimension) and finally
we separate the odd and even dimensional case.

\subsection{Finite volume setup}\label{ss:setup}
Let $M^d$  be a connected, complete Riemannian $d$--manifold without boundary and let $\mathcal{T}_M$ be a tubular tiling of $M^d$ with tiles $\tau_i$. 
Let $T\subset M^d$ be an embedded hypersurface associated with 
$\mathcal T_M$.
As already used in Definition \ref{defn:topological balance}, for $R>0$, let ${\mathcal{B}} _R={\mathcal{B}}_R(P)\subset M^d$ denote the closed, geodesic ball of radius $R$ with center $P \in M^d$:
\begin{equation}\label{eq:ball}
{\mathcal{B}}_R=\{ x \in M^d: \operatorname{dist}(P,x) \leq R\},
\end{equation}
where $\operatorname{dist}(P,x)$ is the geodesic distance between $x$ and $P$. Then we always have
\begin{equation}\label{eq:ball_limit}
\lim _{R \to \infty} {\mathcal{B}}_R=\bigcup _{R>0}{\mathcal{B}}_R=M^d.
\end{equation}
\begin{rem}\label{rem:compact} If $M^d$ is compact, then there exists a finite radius $R=R_0$
(the intrinsic diameter from $P$) such that 
\[
\forall R > R_0: \quad {\mathcal{B}}_R=M^d,
\]
so the limit process
described in equation (\ref{eq:ball_limit}) saturates already at some finite radius $R_0$.
If $M^d$ is noncompact then there is no finite saturation radius $R_0$, the union over all $R>0$
still exhausts the manifold $M^d$.

Thus, equation (\ref{eq:ball_limit}) remains valid both in the compact and the non-compact case.
Henceforth we write equations which are formally valid in both cases but we keep in mind that
the limit process is slightly different. In the case of compact $M^d$ we will assume
that $R>R_0$.
\end{rem}

We denote by $N_A(R)$ and $N_B(R)$ the numbers of $A$-tiles and $B$-tiles,
respectively, that are entirely contained in ${\mathcal{B}}_R$ and we write
\[
N(R):=N_A(R)+N_B(R).
\]
We also define the finite portions of the partitions $A$ and $B$  consisting of tiles entirely contained in ${\mathcal{B}}_R$ and their common tubular surface as
\begin{equation}\label{eq:tr}
A_R := \bigcup_{A_i\subset {\mathcal{B}}_R} A_i,\quad
B_R := \bigcup_{B_j\subset {\mathcal{B}}_R} B_j, \quad T_R:=A_R\cap B_R.
\end{equation}
Thus $A_R$, $B_R$  and $T_R$  are compact. As $R\to\infty$, these whole-tile patches exhaust the manifold and the tubular hypersurface:
$$
\bigcup_{R>0}(A_R \cup B_R)=M^d,\qquad \bigcup_{R>0}T_R=T.
$$
The relative frequencies of $A$- and $B$-tiles can be obtained as
\begin{equation}\label{eq:rf}
p_A(R)=\frac{N_A(R)}{N(R)}, \, p_B(R)=\frac{N_B(R)}{N(R)}.
\end{equation}
We introduce shorthand notations for the Euler characteristics for the tiles and their boundaries as $$\chi_{\mathbf{A},i}=\chi(A_i), \quad\chi_{A,i}=\chi(\partial A_i), $$
and the components of the boundary as
$$\chi_{\hat{A},i}=\chi(\hat{\partial} A_i),  \,
\chi_{\bar{A},i}=\chi(\bar{\partial} A_i), \, \chi_{\mathring{A},i}=\chi(\mathring{\partial}A_i)$$
with corresponding averages 
\begin{equation}\label{eq:tub_eu0}
\chi_{\mathbf{A}}(R)=\frac{1}{N_A(R)}\sum_{i=1}^{N_A(R)}\chi(A_i), \quad \chi_{A}(R)=\frac{1}{N_A(R)}\sum_{i=1}^{N_A(R)}\chi(\partial A_i),
\end{equation}
\begin{equation}\label{eq:tub_eu1}
\chi_{\bar{A}}(R)=\frac{1}{N_A(R)}\sum_{i=1}^{N_A(R)}\chi_{\bar{A},i}, \quad \chi_{\hat{A}}(R)=\frac{1}{N_A(R)}\sum_{i=1}^{N_A(R)}\chi_{\hat{A},i}, \quad \chi_{\mathring{A}}(R)=\frac{1}{N_A(R)}\sum_{i=1}^{N_A(R)}\chi_{\mathring{A},i}.
\end{equation}
For the $B$ partition the notations and averages are analogous.

\subsection{Limits}\label{ss:limits}

By the assumptions of Theorem 1, as \(R\to\infty\), the relative frequencies in (\ref{eq:rf}), the tile averages \(\chi_{\mathbf A}(R),\chi_{\mathbf B}(R)\) in (\ref{eq:tub_eu0}), and the internal-interface averages \(\chi_{\bar A}(R),\chi_{\bar B}(R)\) in (\ref{eq:tub_eu1}) approach finite limits. The remaining averages in (\ref{eq:tub_eu0})–(\ref{eq:tub_eu1}) also approach finite limits, as follows from (\ref{eq:sbdc}) and the pointwise identities (\ref{eq:odd1}),(\ref{eq:odd2}),(\ref{eq:even1}) and (\ref{eq:even2}).
As noted in Remark \ref{rem:compact}, these limits are already saturated at finite $R=R_0$ if $M^d$ is compact and in this case $N_A, N_B, N$
remain finite. In the noncompact case  $N(R)$ will approach infinity.  We can write formally both for the compact and the non-compact case:
\[
p_A:=\lim_{R\to\infty}p_A(R), \quad p_B:=\lim_{R\to\infty}p_B(R), 
\]
\[
\chi_{\mathbf{A}}:=\lim_{R\to\infty}\chi_{\mathbf{A}}(R), \quad \chi_A:=\lim_{R\to\infty}\chi_{A}(R), 
\]
\[
\chi_{\bar{A}}:=\lim_{R\to\infty}\chi_{\bar{A}}(R), \quad \chi_{\hat{A}}:=\lim_{R\to\infty}\chi_{\hat{A}}(R), \quad  \chi_{\mathring{A}}:=\lim_{R\to\infty}\chi_{\mathring{A}}(R).
\]
These limits define the  averages associated with $A$-partition of the tubular tiling $\mathcal{T}_M$ with analogous limits for the $B$--partition.

\subsection{General equations}

Using (\ref{eq:tr}), we can directly apply inclusion--exclusion, using equation (\ref{eq:add})\cite{Hatcher}
to obtain
\begin{equation}\label{eq:IE}
\chi(A_R \cup B_R)=\chi(A_R)+\chi(B_R)-\chi(T_R).
\end{equation}
Our goal is to obtain the three terms on the right hand side of (\ref{eq:IE})
as functions of the averages (\ref{eq:tub_eu0})-(\ref{eq:tub_eu1}).

\subsubsection{Cell boundary equations}\label{ss:cb}

By (\ref{eq:tr}), the part of the tubular hypersurface contained in the whole-tile patch is
\begin{equation}\label{eq:surface_tiling}
T_R= \bigcup_{\substack{A_i\subset{\mathcal{B}}_R\\ B_j\subset{\mathcal{B}}_R}}(A_i\cap B_j).
\end{equation}
Grouping the patches in (\ref{eq:surface_tiling}) either by \(A_i\), or alternatively by \(B_j\), gives the \(A\)-type or the \(B\)-type decompositions of \(T_R\), respectively.
Away from the frontier of the whole-tile patch, the boundary strata of the above tilings, along which the smooth
patches (external interfaces)
$\widehat{\partial}A_i,\widehat{\partial}B_j$ are glued to each other, are
the separation boundaries
$\mathring{\partial}A_i,\mathring{\partial}B_j$.
Considering this fact, and using equations (\ref{eq:triple}) and (\ref{eq:surface_tiling}), we can compute the
Euler characteristic of $T_R$ under the exclusion-inclusion
principle (\ref{eq:add}).

Now we observe that each connected component of the boundary
$\mathring{\partial}A_i$ is a closed smooth $(d-2)$--manifold and, because
of the condition (\ref{eq:triple}), appears in the $R\to\infty$ limit as a boundary
component of exactly two {\dg{tubular tiles of the same phase}}.
Therefore, in the $R\to\infty$ limit,
{\dg{each separation-boundary component}} is covered exactly twice by the
family $\{\mathring{\partial}A_i\}$.

{\dg{The only failure of this cancellation occurs when a tile contained
in $\mathcal B_R$ has a neighboring tile which is not contained in
$\mathcal B_R$. Since tiles are connected, every such omitted neighbor
intersects $\partial\mathcal B_R$, and hence belongs to $\mathcal D_R(P)$.
Thus the discrepancy is supported entirely on the boundary tiles singled
out in Definition \ref{defn:topological balance}. We therefore define the corresponding boundary
errors as the defects in the whole-tile identities:}}

\begin{equation}\label{eq:surface_error_A}
{\dg{
\epsilon_{\mathring A}(R)
:=
\sum_{i=1}^{N_A(R)}
\chi(\widehat{\partial}A_i)
-
\frac12
\sum_{i=1}^{N_A(R)}
\chi(\mathring{\partial}A_i)
-
\chi(T_R),
}}
\end{equation}
{\dg{and analogously}}
\begin{equation}\label{eq:surface_error_B}
{\dg{
\epsilon_{\mathring B}(R)
:=
\sum_{j=1}^{N_B(R)}
\chi(\widehat{\partial}B_j)
-
\frac12
\sum_{j=1}^{N_B(R)}
\chi(\mathring{\partial}B_j)
-
\chi(T_R).
}}
\end{equation}

{\dg{These error terms vanish for all sufficiently large $R$ when $M^d$ is compact. In the noncompact case their normalized size will be estimated in Subsection~B.3.3 using topological balance. With this notation, (\ref{eq:surface_error_A}) can be rewritten as}}
\begin{equation}\label{eq:surface}
\chi(T_R)
=
\sum_{i=1}^{N_A(R)}
\chi(\widehat{\partial}A_i)
-
\frac12
\sum_{i=1}^{N_A(R)}
\chi(\mathring{\partial}A_i)
-
\epsilon_{\mathring A}(R).
\end{equation}
Next we rewrite (\ref{eq:surface}) for the averages as
\begin{equation}\label{eq:euler_T2c}
\chi(T_R)=N_A(R)\left(\chi_{\hat A}{\dg{(R)}}-\frac12 \chi_{\mathring A}{\dg{(R)}}\right)-\epsilon_{\mathring A}(R).
\end{equation}
In the next step we write the Euler characteristics for the boundary of the individual cell.
Using (\ref{eq:add}), (\ref{eq:cellboundary}) and(\ref{eq:sepboundary}), we can write
\begin{equation} \label{eq:sbd0c}
\chi(\partial A_i)=\chi(\hat \partial A_i)+\chi(\bar \partial A_i)-\chi(\mathring{\partial} A_i).
\end{equation}
Rearranging (\ref{eq:sbd0c}) and replacing individual Euler characteristics by their respective averages yields
\begin{equation} \label{eq:sbdc}
\chi_{\hat A}{\dg{(R)}}=\chi_A{\dg{(R)}}-\chi_{\bar A}{\dg{(R)}}+\chi_{\mathring A}{\dg{(R)}}.
\end{equation}
Now we plug (\ref{eq:sbdc}) into (\ref{eq:euler_T2c}): 
\begin{equation}\label{eq:euler_T3c}
\chi(T_R)=N_A(R)\left((\chi_A{\dg{(R)}}-\chi_{\bar A}{\dg{(R)}}+\chi_{\mathring A}{\dg{(R)}})-\frac12 \chi_{\mathring A}{\dg{(R)}}\right)-\epsilon_{\mathring A}(R).
\end{equation}
After rearranging terms and with analogous considerations for the $B$ partition we  obtain
\begin{equation}\label{eq:euler_T3cc}
\chi(T_R)=N_A(R)\left(\chi_A{\dg{(R)}}-\chi_{\bar A}{\dg{(R)}}+\frac{1}{2}\chi_{\mathring A}{\dg{(R)}}\right)-\epsilon_{\mathring A}(R)=N_B(R)\left(\chi_B{\dg{(R)}}-\chi_{\bar B}{\dg{(R)}}+\frac{1}{2}\chi_{\mathring B}{\dg{(R)}}\right)-\epsilon_{\mathring B}(R).
\end{equation}

\subsubsection{Solid cell equations}\label{ss:sc}
Next we consider the $d$-dimensional cells $A_i,B_j$ themselves for which we have
\begin{equation}\label{eq:tilingr3c}
  A_R = \bigcup_{i=1}^{N_A(R)}  A_i, \, B_R=\bigcup_{j=1}^{N_B(R)}  B_j.
\end{equation} 
{\dg{In the interior of $A_R$, every internal interface is shared by two $A$-tiles and is therefore counted twice in the tilewise sum, exactly as in the full tiling. The only discrepancy occurs when an $A$-tile contained in ${\mathcal{B}}_R$ has an $A$-neighbor which is not contained in ${\mathcal{B}}_R$. As in Subsection~B.3.1, every such omitted neighbor intersects $\partial{\mathcal{B}}_R$ and hence belongs to $\mathcal D_R(P)$. We therefore define the boundary error by}}
\begin{equation}\label{eq:solid_error_A}
{\dg{
\epsilon_{\bar A}(R)
:=
\sum_{i=1}^{N_A(R)}\chi(A_i)
-
\frac12\sum_{i=1}^{N_A(R)}\chi(\bar{\partial}A_i)
-
\chi(A_R).
}}
\end{equation}
{\dg{We define $\epsilon_{\bar B}(R)$ analogously. With this notation, (\ref{eq:solid_error_A}) can be rewritten as}}
\begin{equation}\label{eq:solid}
\chi(A_R)
=
\sum_{i=1}^{N_A(R)}\chi(A_i)
-
\frac12\sum_{i=1}^{N_A(R)}\chi(\bar{\partial}A_i)
-
\epsilon_{\bar A}(R).
\end{equation}
{\dg{The analogous identity holds for the $B$ partition. Rewriting (\ref{eq:solid}) for the averages gives}}
\begin{equation}\label{eq:euler_A3c}
\chi(A_R) = N_A(R)\left(\chi_{\mathbf{A}}{\dg{(R)}}-\frac12\chi_{\bar A}{\dg{(R)}}\right)
-\epsilon_{\bar A}(R),
\qquad
\chi(B_R)=
N_B(R)\left(\chi_{\mathbf{B}}{\dg{(R)}}-\frac12\chi_{\bar B}{\dg{(R)}}
\right)-
\epsilon_{\bar B}(R).
\end{equation}

\subsubsection{Unified equation}\label{ss:uni}
In the next step we plug  (\ref{eq:euler_T3cc}) and (\ref{eq:euler_A3c}) into (\ref{eq:IE}) and divide by $N(R)=(N_A(R)+N_B(R))$:
{\tiny{
\begin{equation}\label{eq:u10}
\frac{\chi({\dg{A_R \cup B_R}})}{N(R)}=p_A(R)\left( \chi_{\mathbf A}{\dg{(R)}}-\frac{1}{2}\chi_{\bar A}{\dg{(R)}}\right)- \frac{\epsilon_{\bar A}(R)}{N(R)}+p_B(R)\left( \chi_{\mathbf B}{\dg{(R)}}-\frac{1}{2}\chi_{\bar B}{\dg{(R)}}\right)- \frac{\epsilon_{\bar B}(R)}{N(R)}-p_A(R)\left(\chi_A{\dg{(R)}}-\chi_{\bar A}{\dg{(R)}}+\frac{1}{2}\chi_{\mathring A}{\dg{(R)}}\right)+ \frac{\epsilon_{\mathring A}(R)}{N(R)}.
\end{equation}
\begin{equation}\label{eq:u20}
\frac{\chi({\dg{A_R \cup B_R}})}{N(R)}=p_A(R)\left( \chi_{\mathbf A}{\dg{(R)}}-\frac{1}{2}\chi_{\bar A}{\dg{(R)}}\right)- \frac{\epsilon_{\bar A}(R)}{N(R)}+p_B(R)\left( \chi_{\mathbf B}{\dg{(R)}}-\frac{1}{2}\chi_{\bar B}{\dg{(R)}}\right)- \frac{\epsilon_{\bar B}(R)}{N(R)}-p_B(R)\left(\chi_B{\dg{(R)}}-\chi_{\bar B}{\dg{(R)}}+\frac{1}{2}\chi_{\mathring B}{\dg{(R)}}\right)+ \frac{\epsilon_{\mathring B}(R)}{N(R)}.
\end{equation}
}}
Our  goal is to take the $R \to \infty$ limit for (\ref{eq:u10})-(\ref{eq:u20}).
{\dg{By Definition \ref{defn:topological balance}, for every fixed centre \(P\), the Euler characteristics of the unions of tiles entirely contained in the geodesic balls \({\mathcal{B}}_R(P)\) remain bounded in absolute value for sufficiently large radii. If \(M^d\) is compact, then for sufficiently large \(R\), \(A_R\cup B_R=M^d\) and \(N(R)=N\). If \(M^d\) is noncompact, then \(N(R)\to\infty\), while \(\chi(A_R\cup B_R)\) remains bounded. Hence}}
\begin{equation}\label{eq:lhs}
\lim_{R \to \infty}\frac{\chi({\dg{A_R \cup B_R}})}{N(R)}=\frac{\chi( M^d)}{N},
\end{equation}
{\dg{Our next goal is to prove that all normalized error terms vanish as $R\to \infty$, i.e. we want to prove}}
\begin{equation}\label{eq:err}
\lim_{R \to \infty}\frac{\epsilon_{\mathring A}(R)}{N(R)}=\lim_{R \to \infty}\frac{\epsilon_{\mathring B}(R)}{N(R)}=\lim_{R \to \infty}\frac{\epsilon_{\bar A}(R)}{N(R)}=\lim_{R \to \infty}\frac{\epsilon_{\bar B}(R)}{N(R)}=0.
\end{equation}
{\dg{
By (\ref{eq:surface_error_A}), (\ref{eq:surface_error_B}), (\ref{eq:solid_error_A}), and the \(B\)-analogue of (\ref{eq:solid_error_A}), the error terms measure precisely the failure of the full-tiling cancellations at the frontier of the whole-tile patch. In the interior all interface contributions cancel exactly. Every remaining interface stratum is incident to a tile not contained in \({\mathcal{B}}_R\) but neighboring a contained tile; as observed in Subsections \ref{ss:cb}-\ref{ss:sc}, every such tile belongs to \(\mathcal D_R(P)\). After expanding the boundary contributions by inclusion--exclusion, each error is a linear combination, with fixed coefficients, of interface components and their nonempty pairwise intersections incident to tiles in \(\mathcal D_R(P)\). Hence Definition \ref{defn:topological balance} gives a constant \(C'_P<\infty\), independent of \(R\), such that
\[
\max\left\{
|\epsilon_{\mathring A}(R)|,\,
|\epsilon_{\mathring B}(R)|,\,
|\epsilon_{\bar A}(R)|,\,
|\epsilon_{\bar B}(R)|
\right\}
\le C'_P\,n(R).
\]
Since topological balance also gives \(n(R)/N(R)\to0\),
\[
\frac{|\epsilon(R)|}{N(R)}
\le
C'_P\frac{n(R)}{N(R)}
\longrightarrow0,
\]
for $\epsilon=\epsilon_{\mathring A},\epsilon_{\mathring B},\epsilon_{\bar A},\epsilon_{\bar B}$, which proves (\ref{eq:err}).}}
Using the notation from Subsection \ref{ss:limits}, plugging (\ref{eq:lhs}) and (\ref{eq:err}) into  (\ref{eq:u10})-(\ref{eq:u20}) and swapping the right hand and left hand sides we get:
\begin{equation}\label{eq:u1}
p_A\left( \chi_{\mathbf A}-\frac{1}{2}\chi_{\bar A}\right)+p_B\left( \chi_{\mathbf B}-\frac{1}{2}\chi_{\bar B}\right)-p_A\left(\chi_A-\chi_{\bar A}+\frac{1}{2}\chi_{\mathring A}\right)=\frac{\chi( M^d)}{N}
\end{equation}
\begin{equation}\label{eq:u2}
p_A\left( \chi_{\mathbf A}-\frac{1}{2}\chi_{\bar A}\right)+p_B\left( \chi_{\mathbf B}-\frac{1}{2}\chi_{\bar B}\right)-p_B\left(\chi_B-\chi_{\bar B}+\frac{1}{2}\chi_{\mathring B}\right)=\frac{\chi( M^d)}{N}.
\end{equation}

If we take the sum or the difference of (\ref{eq:u1}) and (\ref{eq:u2}),  we get, respectively
\begin{eqnarray}\label{eq:u_sum}
p_A\left( 2\chi_{\mathbf A}-\chi_A-\frac{1}{2}\chi_{\mathring A}\right)+p_B\left(2 \chi_{\mathbf B}-\chi_B-\frac{1}{2}\chi_{\mathring B}\right)  & = & \frac{2\chi( M^d)}{N}. \\
\label{eq:u_diff}
p_A\left( \chi_A-\chi_{\bar A}+\frac{1}{2}\chi_{\mathring A}\right)-p_B\left( \chi_B-\chi_{\bar B}+\frac{1}{2}\chi_{\mathring B}\right) & = & 0.
\end{eqnarray}

\subsection{Odd dimensions}
If $d=2k+1$ is an odd number, then 
\begin{equation}\label{eq:odd1}
\chi_{\mathring A}=\chi_{\mathring B}=0,
\end{equation} because the separation boundary is the disjoint union of closed, $(d-2)$-dimensional manifolds. 
\begin{lem}
If $d$ is odd then 
\begin{equation}\label{eq:rhs}
\frac{2\chi( M^d)}{N}=0.
\end{equation}
\end{lem}
\begin{proof}
If $M^d$ is compact and $d$ is odd, then we have 
$\chi(M^d)=0$, so the claim is proven. If $M^d$ is noncompact then $N=\infty$.
Since $M^d$ is of finite topological type, $\chi(M^d)$ is finite,
so the claim is also proven for the noncompact case.
\end{proof}
If $d$ is odd, then we can also use 
\begin{equation}\label{eq:odd2}
\chi _A = 2\chi_{\mathbf A}, \quad \chi_B=2\chi_{\mathbf B}.
\end{equation}
Plugging (\ref{eq:odd1}), (\ref{eq:rhs}) and (\ref{eq:odd2}) into (\ref{eq:u_sum}) we get an identity.  Equation (\ref{eq:u_diff}) can be written as
\begin{equation}\label{eq:odd}
p_A\left(2\chi_{\mathbf A}-\chi_{\bar A}\right)-p_B\left(2\chi_{\mathbf B}-\chi_{\bar B}\right)=0.
\end{equation}
\subsection{Even dimensions}
If $d=2k$ is an even number then $\partial A_i, \partial B_j$ are odd-dimensional closed manifolds, so we have 
\begin{equation}\label{eq:even1}
\chi_A=\chi_B=0.
\end{equation} We also observe that the internal interfaces $\bar \partial A_i, \bar \partial B_j$
are compact, odd-dimensional manifolds the respective boundaries of which is the separation boundary $\mathring \partial A$, so we can write 
\begin{equation}\label{eq:even2}
\chi _{\bar A}=\frac{1}{2}\chi_{\mathring A}, \quad \chi _{\bar B}=\frac{1}{2}\chi_{\mathring B}.
\end{equation}
Plugging (\ref{eq:even1}) and (\ref{eq:even2}) into (\ref{eq:u_diff}) yields an identity and (\ref{eq:u_sum}) can be written as
\begin{equation}\label{eq:even}
p_A\left( 2\chi_{\mathbf A}-\chi_{\bar A}\right)+p_B\left( 2\chi_{\mathbf B}-\chi_{\bar B}\right)=\frac{2\chi( M^d)}{N}.
\end{equation}
\subsection{Completing the proof}
Equation (\ref{eq:odd}) proves  Theorem 1 for odd values of $d$ while equation (\ref{eq:even}) proves Theorem 1 for even values of $d$. This completes  the proof of Theorem 1. 
\end{proof}

\section{Substitution table for the examples}
\begin{table}[ht!]
\begin{center}
\begin{tabular}{|c|c|c||c|c|c || c | c || c | c || c | c |}
\hline
 \# & Name         & Fig                          & $d$ & $M^d$          & $\chi(M^d)$ & $p_A$           & $p_B$           & $\chi_{\mathbf{A}}$ & $\chi_{\mathbf{B}}$ & $\chi_{\bar A}$ & $\chi_{\bar B}$ \\
\hline
\hline
1   & Schwarz P   &       5(a)       & $3$ & $\mathbb{R}^3$ & $1$         & $\frac{1}{2}$   & $\frac{1}{2}$   & $1$                 & $1$                 & $6$             & $6$             \\
\hline
2   & Schwarz D   &       1(b)              & $3$ & $\mathbb{R}^3$ & $1$         & $\frac{1}{2}$   & $\frac{1}{2}$   & $1$                 & $1$                 & $4$             & $4$             \\
\hline
3   & Gyroid      &        4(f)                & $3$ & $\mathbb{R}^3$ & $1$         & $\frac{1}{2}$   & $\frac{1}{2}$   & $1$                 & $1$                 & $3$             & $3$             \\
\hline
\hline
4   & Fermi Cu    &      6              & $3$ & $\mathbb{R}^3$ & $1$         & $\frac{1}{2}$   & $\frac{1}{2}$   & $1$                 & $0$                 & $8$             & $6$             \\
\hline 
5   & Honeycomb   &      2(a)       & $2$ & $T^2$          & $0$         & $\frac{n}{n+1}$ & $\frac{1}{n+1}$  & $1$                 & $-n$                & $0$             & $0$             \\
\hline
6   & Pollen      &  3(b)       & $2$ & $S^2$          & $2$         & $\frac{n}{n+1}$ & $\frac{1}{n+1}$  & $1$                 & $2-n$               & $0$             & $0$             \\
\hline
7   & RW          &  7           & $3$ & $S^3$          & $0$         & $\frac{n}{n+m+1}$ & $\frac{m+1}{n+m+1}$ & $1$            & $\frac{2m+n-N}{m+1}$ & $0$             & $\frac{2m(2-k)}{m+1}$        \\
\hline
\end{tabular}
\caption{Substitution table for the examples illustrated in the article.}\label{tab:fermi}
\end{center}
\end{table}


\bibliographystyle{unsrt}
\bibliography{tub1}

\end{document}